\documentclass[twoside]{article}
\usepackage[accepted]{aistats2024}
\usepackage{xcolor}
\usepackage{enumitem}
\usepackage{multirow, makecell}
\usepackage{bm}
\begin{document}
\twocolumn[
\aistatstitle{Simulation-Based Stacking}
\aistatsauthor{Yuling Yao$^*$   \And Bruno Régaldo-Saint Blancard$^*$  \And Justin Domke}
  \runningauthor{Yuling Yao, Bruno Régaldo-Saint Blancard,  Justin Domke}
\aistatsaddress{
{Flatiron Institute, New York$^*$}
\And  
{\hspace{3.5cm}University of Massachusetts Amherst} }
]

\newcommand{\jd}[1]{{\color{olive}[JD: #1]}}
\newcommand{\ve}[1]{{\bm{#1}}}
 
\begin{abstract}
Simulation-based inference has been popular for 
amortized Bayesian computation. 
It is typical to have more than one posterior approximation, from different inference algorithms, different architectures, or simply the randomness of initialization and stochastic gradients. 
With a consistency guarantee, 
we present a general posterior stacking framework to make use of all available approximations. Our stacking method is able to combine densities, simulation draws, confidence intervals, and moments, and address the overall precision, calibration, coverage, and bias of the posterior approximation at the same time. We illustrate our method on several benchmark simulations and a challenging cosmological inference task.   
	\end{abstract}
		\vspace{-.08in}
	
\section{Introduction}

Simulation-based inference (SBI) has been widely used in scientific computing including biology, astronomy, and cosmology \citep[e.g.,][see Appx.~\ref{app:back} for background]{cranmer2020frontier, gonccalves2020training,  dax2021real, hahn2022simbig}. Instead of an explicit likelihood function, SBI only requires a forward model that generates simulated observations given parameters. Despite its popularity, there has been a growing concern about the sampling quality of SBI: how accurate the inference is compared with the true posterior.  Simulation-based calibration \citep[SBC,][]{talts2018validating} diagnoses posterior miscalibration. Given sufficient data, it will typically reject the null hypothesis because all computations are approximations. But the goal of computation calibration is not to reject. Given some imperfect inferences, what is next? This paper develops a stacking approach to aggregate these miscalibrated SBI outcomes, such that the aggregated inference is closer to the true posterior.

Moreover, non-mixing computation is prevalent in SBI. 
For one fixed inference task, practitioners often obtain many different posterior inference results  
because of different neural network architectures and hyperparameters, because the posterior itself can be multimodal and it is hard for one inference run to traverse across all isolated modes, or because it is cheaper for modern hardware to run many short-and-crude simulation-based inferences in parallel.
For illustration, Fig.~\ref{fig:rankssimBig} visualizes the divergent computing results in a challenging cosmology problem (SimBIG, Section \ref{sec:exp}). 
After varying hyperparameters of neural posterior estimators (normalizing flows) on seemingly reasonable ranges, we obtain up to 1,000 posterior inferences. 
The rank statistics of one parameter across four inference runs display various miscalibration types, indicating biases, over- and under-confidence.  The expected log densities (the log data likelihood, or the negative loss functions) across 1,000 inferences vary by a range of 1.7 nats.

\begin{figure}[!b]
    \centering
    \vspace{-1em}
    \includegraphics[width=\hsize]{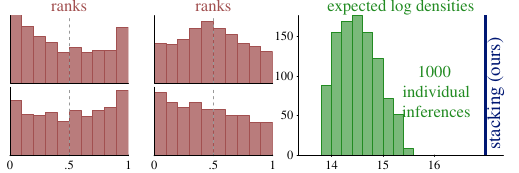}   
     \vspace{-2em}
    \caption{We run 1,000 neural posterior inferences in a challenging cosmology model. The rank histograms of one parameter reveal different types of miscalibration in four runs. The expected log densities of the 1,000  inferences vary by 1.7 nats, while stacking from this paper  improves the best approximation by 1.4 nats on holdout simulations.}
    \label{fig:rankssimBig}
\end{figure}

To handle non-mixing posterior inferences, a remedy is to pick one inference, but the selection procedure itself is noisy. Even if we correctly pick the best inference, not exploiting suboptimal inferences wastes computation, inflates the Monte Carlo variance, and reduces estimation efficiency.
Another option is to ``average over'' all inferences. But uniform weighting is generally not optimal and could be a bad idea when there are many bad inferences. Moreover, there are various ways to aggregate inferences, such as taking a linear combination of posterior densities (mixture), a combination of posterior samples, or a combination of confidence intervals.  
At the same time, posterior approximation can have various goals, such as that the approximate posterior density be close to the truth in some divergence metric, that posterior ranks should be calibrated, that the posterior mean is unbiased, or that the posterior confidence interval attains nominal coverage. If the inference is exact, all of these goals will match. But in reality, most computations are approximate, leading to tension between these  goals.



This paper develops a stacking approach to combine
multiple simulation-based inferences to improve distributional approximation. As a meta-learning procedure, instances of this framework take many individual inferences as input and output a ``stacked'' distribution to better approximate the true posterior in a given metric.  
To make the stacked distribution more flexible, we design three \emph{aggregation forms}: density mixture, sample aggregation,  and interval aggregation.  To facilitate various user-specific utility of distribution approximating, we design \emph{objective functions} on  Kullback–Leibler divergence, rank-based calibration, coverage of posterior intervals, and mean-squared error of moments. Any product of an aggregation form and an objective function renders a stacking method.  We develop five practical posterior stacking methods in Section \ref{sec:method}.  We organize them in a general  framework in Section \ref{sec:theory}, where we further prove that the stacked SBI posterior is asymptotically guaranteed to be the closest to the true posterior distribution in the assigned divergence metric.   We recommend hybrid stacking to balance different perspectives of distribution approximation. 
In Section \ref{sec:exp}, we illustrate the implementation of our methods in simulated and real-data examples, which involves a cosmology problem regarding galaxy clustering. We discuss related methods and further directions in Section \ref{sec:discuss}.

\section{\`A la carte stacking }\label{sec:method}
Throughout the paper we work with the general SBI setting where the goal is to sample from a posterior density $p(\theta|y)$ with a potentially intractable likelihood. The parameter space $\Theta$ is a subset of $\R^d$, and no assumption is needed for the data space.  
We create a  size-$N$ joint simulation table  $\{(\theta_n, y_n)\}_{n=1}^N$ by first drawing parameters $\theta_n$ from the prior distribution $p(\theta)$ and one realization of data $y_n$ from the data-forward model $p(y|\theta_n)$. 
From this simulation table, we run $K$ simulation-based inferences coming from various algorithms or architectures. Given data $y_n$, the $k$-th inference,  $k=1, \dots, K$,  returns a learned posterior density $q_{k}(\theta|y_n)$, and $S$ posterior samples $\tilde \theta_{kns}$ from $q_{k}(\theta|y_n)$, $s=1,\dots, S$.
The goal of stacking is to construct an ensemble of SBI posteriors, making use of the inferred approximation densities  $q_{k}(\cdot|y)$ or/and the sample draws $\tilde \theta$, such that the aggregated approximation is as close to the true posterior $p(\theta|y)$ as possible. See Figure \ref{fig_workflow} for an illustration.

\begin{figure}
 \vspace{-0.3em}
 \begin{tikzpicture} 
 \footnotesize
  \node (theta) at (0,0) {\color{blue}  $\theta_n$};

  \node (theta2) at (0.4,1.2) {\color{blue} $N$ samples};
  \node (theta3) at (0.4,0.85) {\color{blue} from prior};
  \node (y2) at (1,-0.38) {\color{green}  data};
  \node (y3) at (3.6,1.2) {\color{darkblue}  $K$ inferences with $S$ samples};
  \node (y4) at (6.75,0.85) {\color{red}  stacked posterior};

  \node (y) at (1,0) {\color{green} $y_n$};
  \node (theta1) at (1.9,0.8) {};
  \node (theta2) at (1.9,0.4) {};
    \node (theta32) at (1.9,0) {};
  \node (theta3) at (3.6,0.15) {\begin{tabular}{c}
${\color{darkblue} \tilde{\theta}_{1n1}, \dots, \tilde{\theta}_{1nS}} \sim q_1(\cdot | y_n) $\\
 ${\color{darkblue}\tilde{\theta}_{2n1}, \dots, \tilde{\theta}_{2nS} }\sim q_2(\cdot | y_n)$\\
$\cdots$  \\
${\color{darkblue}\tilde{\theta}_{Kn1}, \dots, \tilde{\theta}_{KnS} }\sim q_K(\cdot | y_n)$ \\
\end{tabular} };
  \node (theta4) at (1.85,-0.5) {};
  \draw[->] (theta) -- (y);
  \draw[->] (y) -- (theta1);
  \draw[->] (y) -- (theta2);
  \draw[->] (y) -- (theta32);
  \draw[->] (y) -- (theta4);
\draw [thick, decorate, decoration = {brace}] (5.4,0.8) --  (5.4,-.6);
 \node (midpoint1) at (5.4,0.1) {};
  \node (midpoint2) at (6,0.1) {};
   \draw[ double, ->] (midpoint1) -- (midpoint2);
 \node (final2) at (6.9,-0.1) {$\sim q^*(\cdot | y_n)$};
 \node (final1) at (6.9,0.26) { \color{red} $\theta^*_{n1}, \dots, \theta^*_{nS}$}; 
  \node (index) at (7.14,-0.64) {\color{blue} \tiny $n=1, \dots, N$};
  \node (N) at (6.45,-0.9) {\color{olive}   $N$:  {\tiny \#   prior simulations.}};
   \node (K) at (0.8,-0.9) {\color{olive}   $K$: {\tiny  \# of inferences.}};
\node (S) at (3.45,-0.9) {\color{olive}   $S$: {\tiny  \#  post. draws.}};

\end{tikzpicture}
 
 \vspace{-0.4em}
\caption{Stacking has two stages. The first is to sample $N$ prior simulations and run $K$ inferences methods, $q_k( \cdot|y), 1\leq k \leq  K$.  The second stage learns a stacked posterior $q^*(\cdot |y)$ to better approximate the true posterior $p(\theta|y)$.  We design stacking to facilitate different distribution combination forms and learning objectives. \vspace{-0.5em} }\label{fig_workflow}
\end{figure}
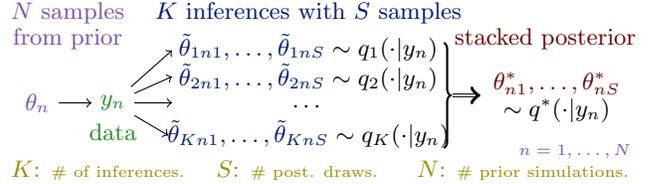

This section develops five practical posterior stacking algorithms. We defer the related theory propositions 1 to 6 and proofs in Appx.~\ref{app:theory}.

\subsection{Density mixture for  KL divergence}\label{sec:mixture_KL}
Perhaps the most natural form to combine posterior densities is to take a linear density mixture
\begin{align}
    \label{eq:simplex}
    \vspace*{-0.5em}
    q^{\mix}_{\w}(\theta|y) \coloneqq \sum_{k=1}^K w_k q_k(\theta|y),
\end{align}
where the weight $\w=(w_1, \dots, w_K)$ lies in a simplex:
$$\w \in \s_K\coloneqq\{ 0 \leq w_{k} \leq 1, ~ \sum_{k=1}^K w_k = 1\}.$$
To find the optimal weights $\w$, we seek to maximize the log predictive density of $q^{\mix}_{w}(\theta|y)$, averaged over simulations $\{\theta_1\dots, \theta_N\}$,
\begin{align}\label{eq:mixture_obj}
    \vspace*{-0.5em}
\hat \w=  \arg\max_{\w \in  \s_K}  \sum_{n=1}^N 
\log \left( \sum_{k=1}^K w_k q_k(\theta_n| y_n)  \right). 
\end{align} 
The expected log  density is connected to the Kullback–Leibler (KL) divergence. If the size of the simulation table $N$ is big enough, then 
up to a constant, the objective function in \eqref{eq:mixture_obj} divided by $N$ converges to the negative  conditional KL divergence\footnote{Standard notation \citep{cover1991elements} for conditional divergence  is  $\KL(p(\theta \vert y) \Vert q(\theta \vert y))  \coloneqq  \E_{p(y,\theta)} \log ({p(\theta \vert y)}/{q(\theta \vert y))}$, \emph{not} divergence of conditionals.}, $\KL( p(\theta|y ),  q^{\mix}_{\w}(\theta|y ) )$; see Prop.~\ref{thm:KL} in Appendix.

\paragraph{Local mixture.}\label{sec:local_mixture}
All computations are wrong, but some are useful \emph{somewhere}.
It is easy to locally adapt the weight $\w(y)$ as a function of data $y$ and output a simplex.  In practice, let $\alpha_1(y)=0$ and for $k>1$, let $\alpha_k(y)$ be the output of a neural network with its own parameters, and set $w_k(y)=\exp (\alpha_k(y))/\sum_{k=1}^K \exp (\alpha_k(y))$. The locally combined posterior is 
$
q^{\mix}_{\w}(\theta|y) \coloneqq \sum_{k=1}^K w_k(y) q_k(\theta|y). 
$ Stacking maximizes its log predictive density  average over simulations $\max_{\w: y\to \s_K} U(\w) \coloneqq \sum_{n=1}^N \log q^{\mix}_{\w}(\theta_n|y_n)$.


Despite the simplicity, there are two reasons to extend this mixture-stacking-to-max-log-density. First, the mixture has a limited degree of freedom $\w\in\s_K$, which limits the flexibility of the stacked posterior. 
Second, even in the mixture form, the log-density-based learning \eqref{eq:mixture_obj} does not make use of the existing simulation draws $\{\tilde \theta_{kns}\}$, which intuitively can offer more information.  The next subsection uses  simulation draws.



\subsection{Density mixture for rank calibration}\label{sec:cali}
\paragraph{From rank-based calibration to rank-based divergence.} 
The rank statistic gives an alternate measurement of posterior approximation quality.
For simplicity, first assume the parameter space $\Theta$ is one-dimensional. 
In the $k$-th inference, we compute the rank statistic (or the $p$-value) of the prior draw $\theta_{n}$ among its paired posterior $q_k (\theta | y_n)$, i.e.,  
\begin{equation}\label{eq:rank}
r_{kn} \coloneqq \frac{1}{S}\sum_{s=1}^S \mathds{1} \{ \tilde \theta_{kns}\leq \theta_{n}  \}.  
\end{equation}
If  the inference is calibrated, $q_k(\cdot | y)= p(\cdot | y)$, then $r_{kn}$ is uniformly distributed over the grid $\{0,1/S, \cdots, 1\}$. \citet{talts2018validating} used this fact to design a
rank-based hypothesis testing to test whether the posterior is exact.
Taking one step further, we quantify the degree of miscalibration.
Given any two conditional distributions $p(\theta|y)$ and $q(\theta|y)$, we define $D_{\mathrm{rank}} (p, q)$, a rank-based generalized divergence metric  as follows. Let  $D_0(X_1, X_2) \coloneqq \int_{0}^1 | \Pr(X_1\leq z) -\Pr(X_2 \leq  z ) |^2 dz$ be a distance between two random variables $X_1$ and $X_2$ on $[0,1]$. $D_0$ compares distributions in cumulative distribution functions (CDFs), which is of the Cramér–von Mises type \citep{cramer1928composition},  and coincides with the continuous ranked probability score \citep{matheson1976scoring}. Consider CDF transformations:  $F_{p}(\theta|y) \coloneqq \int_{-\infty}^{\theta} p(\theta^\prime|y)d\theta^\prime$ and  $F_{q}(\theta|y) \coloneqq  \int_{-\infty}^{\theta} q(\theta^\prime|y)d\theta^\prime$. When $(\theta, y)$ are distributed from $p(\theta, y)$,  $F_{p}(\theta|y) $ and $F_{q}(\theta|y) $ are two random variables on [0,1]. Then define 
\begin{equation}\label{eq_rankdiv}
D_{\mathrm{rank}}(p, q)\coloneqq  D_0\left(F_{p}(\theta|y), F_{q}(\theta|y)\right), ~~(\theta,y) \sim p. 
\end{equation}
This metric $D_{\mathrm{rank}}(p, q)$ is non-negative and its zero is attained when $q(\theta|y ) = p(\theta|y)$ almost everywhere (Prop.~\ref{thm:rank_div}) and hence a generalized divergence. 
Our defined $D_{\mathrm{rank}}(p, q)$ is appealing since it admits a straightforward empirical estimate, 
$D_{\mathrm{rank}}(p(\cdot|y), q_k(\cdot|y))\approx
\int_0^1 \left(\hat{F}_{r_k}(t) -  t\right)^2dt
$, where $\hat{F}_{r_k}(t) = \frac{1}{N}\sum_{n=1}^N \mathds 1 (r_{kn}\leq  t)$ is the empirical CDF of ranks.  
This integral is computed in a closed form (Appx.~\ref{app:practical}).
Moreover, this sample estimate is differentiable on $r_{kn}$ almost everywhere, in contrast to the familiar Kolmogorov–Smirnov test which takes the supremum norm or the Chi-squared test which requires binning.



\paragraph{Rank in the mixture stacking is linear.}
With $K$ approximate inferences,
we still study a mixture posterior ${q^{\mix}_{\w}(\theta|y_n)= \sum_{k=1}^K w_k q_k (\theta | y_n)}$ and want it to be as correct as possible under rank-based calibration.
Conveniently, the rank of $\theta_{n}$ in any $\w$-weighted mixture has an explicit expression using  individual ranks, 
\begin{equation}\label{eq:linear_rank}
\vspace{-0.5em}
r^{\mix}_{n} \coloneqq \sum_{k=1}^K w_k r_{kn}. 
\end{equation}
Let $\theta^{\mix}$ denotes a random variable with the law  ${q^{\mix}_{\w}(\cdot | y_n) =\sum_{k=1}^K w_k q_k (\theta | y_n)}$. For any fixed $\w$ and $\theta_{n}$, as $S\to \infty$, this $r^{\mix}_{n}$ is a consistent estimate of the mixture CDF, i.e., 
$r^{\mix}_{n} | \mathbf{w}, \theta_{n} \to  
\Pr ( \theta^{\mix} \leq   \theta_{n} )$;  
see Prop.~\ref{thm:rank_linear}.
The linear-additivity \eqref{eq:linear_rank} of the rank statistics can be extended to the local mixture, where the rank of the local mixture posterior is $r^{\mix}_{n} \coloneqq \sum_{k=1}^K w_k(y_n) r_{kn}$.

\paragraph{Stacking for rank calibration.} 
With the rank-based divergence $D_{\mathrm{rank}}$ and the closed form mixture rank $r^{\mix}_{n}$ in \eqref{eq:linear_rank}, 
we are now ready to run a calibration-aware stacking. We seek to minimize the rank-based divergence $D_{\mathrm{rank}}
(p(\cdot|y), q^{\mix}_{\w}(\cdot| y))$ by
\begin{align}\label{eq:rank_obj}
\hspace{-0.6em}
\min_{\w \in \s_K} \int_0^1  \left(\hat{F}_{r^{\rm mix}}(t) -  t\right)^2 dt,    
\end{align}
where $\hat{F}_{r^{\rm mix}}(t) = \frac{1}{N}\sum_{n=1}^N \mathds 1  \left(\sum_{k=1}^K w_k r_{kn} \leq  t\right)$.
The integral has a closed-form expression and hence the optimization is straightforward (Appx.~\ref{app:practical}). 



In addition to $D_0$ that matches the CDFs of the stacked ranks $r^{\mix}$ and the uniform distribution, we can also match their moments, such as to minimize the squared errors,  $(\sum_{n=1}^N \sum_{k=1}^K w_k r_{kn}/N - 1/2)^2$ and $(\sum_{n=1}^N\sum_{k=1}^K \log(w_k r_{kn})/N +1)^2.$ 
Along with \eqref{eq:rank_obj},  these rank-based stacking objectives encourage uniform ranks. As an orthogonal complement to log density stacking \eqref{eq:mixture_obj}, rank-based stacking depends on the approximate inferences $\{q_k(\cdot|y)\}$ through and only through their \emph{sample draws}, not \emph{densities}, which is especially suitable when we cannot evaluate the inferences densities such as in short MCMC and GAN.

In reality $\Theta$ is not one dimensional. Similar to the practice of SBC, either we pick a one-dimensional summary statistic $f(\theta, y)$, compute its rank  $r_{kn} \coloneqq \sum_{s=1}^S \mathds{1} \{f(\tilde \theta_{kns}, y_n) \leq  f(\theta_{n}, y_n)  \} / S$, and run one-dimensional stacking \eqref{eq:rank_obj} for targeted calibration, or we compute ranks for each dimension separately and sum up the objective function  \eqref{eq:rank_obj} on every dimension.

\subsection{Sample  stacking for discriminative calibration}\label{sec:sample}
So far we only consider density mixtures. It is also natural to work with samples directly.  For a given 
$n$ and any $s$, $(\theta_{1ns}, \theta_{2ns}, \dots, \theta_{Kns} )$ are posterior draws from $K$ inferences for the same inference task $p(\theta| y_n)$.
For example, a linear additive model  stacks $K$ approximate samples into one aggregated draw:
\begin{equation}\label{eq:sample_stack_linear}
\theta^*_{ns} = \w_0 + \mathbf{w}_1 \tilde \theta_{1ns} + \dots +  \mathbf{w}_K \tilde  \theta_{Kns},
\end{equation}
where the parameter $\mathbf{w}_0 \in \R^d, \mathbf{w}_k \in \R^{d\times d}, k\geq 1$.

We want the aggregated sample $\{\theta^*_{n1}, \cdots, \theta^*_{nS}\}$ to be a better sample approximation of the posterior $p(\theta| y_n)$. 
To measure the sampling quality in SBI, we adopt discriminative calibration \citep{yao2023discriminative}: if no classifier can distinguish between $\{(\theta_n, y_n)\}$ and  
$\{(\tilde \theta^*_{n1}, y_n), \cdots, (\tilde \theta^*_{nS}, y_n)\}$, then the stacked inference is accurate. Formally, for the $n$-th simulation run, we create $S+1$ binary classification examples. The first example is $\phi=(\theta_n, y_n)$ with label $z=1$, and the ${(s+1)\mathrm{-th}}$ example is 
$\phi=(\tilde \theta^*_{ns}, y_n)$ with label $z=0$. Looping over $1\leq n \leq N$  yields $N(S+1)$  examples $\{(\phi, z)\}$, in which $\phi$ depends on stacking weights $\w$ via $\tilde \theta$.
Denote $P(z|\phi)$ to be a probabilistic classifier that predicts label $z$ using feature $\phi$, where we reweight the classification loss function to balance two classes. Sample-based stacking solves a minimax optimization: 
\begin{equation}\label{eq_minimax}
    \hat \w =  \arg\min_{\mathbf{w}}  \max_{P} \sum_{i=1}^{N(S+1)} \log P (z_i | \phi_i).
\end{equation}
Let $q^*_{\w}(\theta|y)$ be the distribution of the stacked samples~\eqref{eq:sample_stack_linear}. 
As $N\to\infty$, this stacked $q^*_{\hat w}(\theta|y)$ minimizes the  Jensen-Shannon (JS) divergence between any $q^*_{\w}(\theta|y)$  and true posterior $p(\theta|y)$. See Prop.~\ref{thm:JS}.



\subsection{Interval stacking for conformality }\label{sec:interval}
Often the focus of Bayesian inference is to correctly quantify the uncertainty in one of a few parameters for downstream tasks such as hypothesis tests or decision theory tasks.
For simplicity, again assume a one-dimensional parameter space of interest $\Theta$ (otherwise, stack each dimension separately).
Given any $y_{n}$, in the $k$-th inference, let $l_{kn}$ and $r_{kn}$ be the left and right interval endpoint of the $(1-\alpha)$ central confidence interval in $q_k(\theta | y_n)$, which typically is computed via the $ \alpha/2$ and $ 1-\alpha/2$ sample quantiles in $\{\tilde \theta_{kns}: 1\leq s \leq S\}$. If the inference is calibrated or conformal,  the coverage probability of this interval should be at least $(1-\alpha)$ under the true posterior $p(\theta|y_n)$.

To achieve appropriate coverage,  we stack  $K$ individual posterior intervals $\{(l_{kn}, r_{kn}): k\leq K\}$ to produce an aggregated interval $(l^*_n, r^*_n)$.
We adopt a simple linear form with the stacking parameter $\w\in \R^{2K}$,
\begin{equation}\label{eq_interval_stack}
l^*_n = \sum_{k=1}^K  w_k l_{kn}, ~~
r^*_n = \sum_{k=1}^K  w_{k+K} r_{kn}.     
\end{equation}
Besides the correct coverage, we also want the posterior interval to be as narrow as possible to enhance estimation efficiency. The trade-off between coverage and efficiency has been studied in the prediction literature \citep{gneiting2007strictly}.  We design the following interval score stacking, which encourages the coverage and penalizes the length:
\begin{align}
\min_{\w}  &\sum_{i=1}^N U(r^*_n,l^*_n, \theta_n ), \notag \\
\text{where } U(r,l, \theta ) \coloneqq 
&(r - l) 
+ \frac{2}{\alpha} (l - \theta) \mathds{1} (\theta< l) \notag \\
&+ \frac{2}{\alpha} (\theta - r) \mathds{1} (\theta> r).     \label{eq:interval_score}
\end{align}
The stacked interval $(r^*_n,l^*_n)$ asymptotically provides the optimal posterior quantile estimation---As $N$ approaches  $\infty$, the unique minimizer to the 
loss function above is when the stacked interval $(r^*(y), l^*(y))$ is identical to the exact pair of the true $ \alpha/2$ and $ 1-\alpha/2$ quantiles in  $p(\theta|y)$ for almost every $y$ (Prop.~\ref{thm:interval}). 


Unlike the density mixture or sample addition, the stacked interval \eqref{eq_interval_stack} is reduced-form: we do not specify how to sample from the stacked distribution. Our interval stacking  \eqref{eq:interval_score} is a semiparametric approach in which any aspect of the posterior distribution other than the quantile is treated as a nuisance parameter.

\subsection{Moment stacking for unbiasedness/MSE}\label{sec:Moment}
Perhaps the posterior mean and covariance remain the two most important summaries of the posterior distribution. We can directly stack these summaries from $K$ approximations. 
In the $k$-th individual approximation, the sample mean of $q_k (\cdot|y_n)$ is $\mu_{kn} \coloneqq \sum_{s=1} \tilde \theta_{kns}/S$. 
In the mixture stacking,  the posterior mean of $q_{\w}^{\mix}(\cdot|y_n) = \sum_{k} w_k q_k(\cdot | y_n)$ is $\mu^*_n(\w)=\sum_{k} w_k \mu_{kn}$. For sample-based stacking \eqref{eq:sample_stack_linear}, the posterior mean is similar $\mu^{\mix}_n(\w)= \w_0 +  \sum_{k=1} \w_k \mu_{kn}$.
In either case, we can optimize stacking weights to 
match the first moment, 
\begin{equation}\label{eq_match_mean}
\min_{\w}  \sum_{n=1}^N ||\mu^*_n(\w) - \theta_n  ||^2. 
\end{equation}
Likewise, the sample covariance of the $k$-th posterior inference given $y_n$ is $V_{kn} \coloneqq  \sum_{s=1}^S ||\tilde \theta_{kns}-\mu_{kn}||^2 /S$. Using the 
law of total variation, 
the covariance of the mixture $q_{\w}^{\mix}(y_n)$ 
is $V^{\mix}_n(\w)=\sum_{k} \w_k  V_{kn} + \sum_{k} \w_k  || \mu_{kn} - \bar \mu_{n}||^2$, where  $\bar \mu_{n}= \sum_{k} w_k  \mu_{kn}$.
We design moment stacking to minimize the following negative-oriented objective function which matches the two moments at the same time:
\begin{align}
\min_{\w} &\sum_{n=1}^N U(q^{\mix}_w(\cdot |y_n), \theta_n),
\end{align}
where
\begin{align}
U(q^{\mix}_w(\cdot |y_n),\,&\theta_n) \coloneqq \log \mathrm{det} V^{\mix}_n(\w) \notag \\
&+||\mu^{\mix}_n(\w) - \theta_n||^2_{(V^{\mix}_n(\w))^{-1}} \label{eq:moment}  
\end{align} 
where $||u||^2_{\Gamma} \coloneqq u^T \Gamma u$  is the weighted norm. As $N\to \infty$, the minimal of the loss function is achieved 
if and only for almost sure $y$, the stacked mean
$\mu^{\mix}_{\w}$ 
and covariance $V^{\mix}_{\w}$  exactly matches the 
$\E(\theta|y)$ and $\mathrm{Var} (\theta|y)$ in the the true posterior (Prop.~\ref{thm:moment}).

\section{Unified posterior stacking}\label{sec:theory}

In this section, we give a unified presentation of the previous five stacking methods in Sec.~\ref{sec:method}. We observe that they principally vary along two dimensions:

\paragraph{I. What is the combination form?}
Consider $K$ conditional distributions $q_1(\cdot | y), q_2(\cdot | y), \cdots, q_K(\cdot | y)$ that have support on $\Theta$ and represent approximations of the posterior distribution given the same $y$. We define a \emph{combination form} $\Phi$ that maps  $K$ conditional distributions into one stacked conditional distribution:
\begin{equation}\label{eq:stackForm}
\Phi: ~~\{q_1(\cdot | y), q_2(\cdot | y), \cdots, q_K(\cdot | y)  \} \to  q_{\w}^*(\cdot | y).      
\end{equation}
where $\w$ is the stacking parameter. The output  $ q_{\w}^*(\cdot | y)$ should be understood as a conditional distribution, which does not necessarily require an explicit density.

\paragraph{II. What is the objective function?} 
To evaluate how well the stacked approximate inference $q_{\w}^*(\cdot | y)$ approximates the true posterior distribution, we need a utility function. We formulate this sampling valuation into a conditional prediction evaluation task. The simulation table gives paired simulations $(y, \theta)$ from their joint distribution $p(y, \theta)$, such that for any $y$, the paired $\theta$ can be viewed as an independent draw from the unknown posterior $p(\theta|y)$. The stacked inference  $q_{\w}^*(\cdot | y)$ is a  conditional distribution of $\theta$ given $y$. 
A \emph{scoring rule} \citep{gneiting2007strictly} is a bivariate function that compares  any $\Theta$-supported distribution $q(\cdot)$ and a realization $\theta$,
\begin{equation}\label{eq:scoreRule}
U: (q, \theta) \to \R, ~~\theta \in \Theta.    
\end{equation}
Table \ref{tab:listOfMethod} summarizes where our developed methods fit along the combination forms and utility functions.  The table is sparse: it is challenging to fill the remaining entries. For example, the confidence interval of the mixture or the density of sample aggregation is  intractable. We now explore general posterior stacking with an arbitrary  combination form and utility.


{

 \begin{table}
 \footnotesize
 \hspace{-0.8em}
\begin{tabular}{l| c c c c c }
  &log score &rank  &JS div.  &interval
   &moment  
  \\
     \hline
 mixture & \ref{sec:mixture_KL}  &  \ref{sec:cali}  &    & &\ref{sec:Moment} \\  
 sample   &  &   & \ref{sec:sample}  &  &\ref{sec:Moment} \\
 interval    &  &  &   & \ref{sec:interval} \\
\end{tabular}
\caption{Table of five stacking methods in relation to combination forms and utility functions.
 We have used three combination forms: (a)~density mixture, (b)~sample aggregation, and (c)~interval aggregation, and five utility functions (goals):
(i) log score \eqref{eq:mixture_obj}, 
(ii)~rank based calibration \eqref{eq_rankdiv}, (iii)~ negative Jensen-Shannon  divergence, (iv)~interval coverage  \eqref{eq:interval_score}, and (v) moment score \eqref{eq:moment}.
Adding up utility functions from one row forms a hybrid stacking.
}\label{tab:listOfMethod} 
\vspace*{-1em}
\end{table}
}

\paragraph{Learning and consistency.}
We need two conditions to produce a valid posterior stacking method. First, we need to evaluate the score of the stacked distribution, 
$U(q_{\w}^*(\cdot | y_n), \theta_n)$. Second, the 
scoring rule $U$ in \eqref{eq:scoreRule} needs to be proper, i.e., for any two $\theta$-distributions $p$ and $q$, 
\begin{equation}\label{eq:proper}
\int_{\Theta} U(q, \theta) p(\theta) d\theta \leq  \int_{\Theta} U(p, \theta) p(\theta) d\theta.
\end{equation}
These two conditions produces a stacking method. We combine $K$ posterior inferences with the form \eqref{eq:stackForm}, and fit stacking parameters $\w$ via a sample M-estimate, 
\begin{equation}\label{eq_stack}
\hspace{-0.5em}
 \hat \w = \arg \max \sum_{n=1}^N U(q_{\w}^*(\cdot | y_n), \theta_n).
\end{equation}
The expectation 
$\E_{p(y, \theta)} U(q_{\hat \w}^*(\cdot | y), \theta)$  is the average utility function of the stacked posterior.
Unlike a typical Bayesian prediction evaluation task where there are a large number of observations from one fixed data generating process, here for each fixed $y_n$ we only have one draw $\theta_n$ from the true posterior $p(\theta|y)$. 
As reassurance, the next proposition shows that stacking estimate $\hat \w$ from \eqref{eq_stack} is asymptotically optimal.
\begin{proposition}\label{thm:consi} If the score $U$ is proper, then for any $\epsilon>0$ and any given $\w^{\prime}$, as $N\to \infty$,  
$\Pr \left( \E_{p(y, \theta)} U(q_{\hat \w}^*(\cdot | y), \theta) \leq    \E_{p(y, \theta)} U(q_{\hat \w^{\prime}}^*(\cdot | y), \theta)+\epsilon \right) \to 1.$
\end{proposition}

This M-estimator \eqref{eq_stack} covers all stacking procerures in this paper except for the rank stacking \eqref{eq:rank_obj}. In a companion paper \citep{yao2024stacking}, we derive more theories and prove the convergence rate and asymptotic normality of the rank-stacking estimator \eqref{eq:rank_obj} and the M-estimator \eqref{eq_stack}.

A proper scoring rule  $U$ produces a generalized divergence by defining $D_U(p, q) = \int_{\Theta} U(p, \theta) p(\theta) d\theta - \int_{\Theta} U(q, \theta) p(\theta) d\theta.$
Proposition~\ref{thm:consi} implies that the stacked approximation  $q_{\hat \w}^*(\cdot | y)$ is asymptotically optimal as its divergence $D_U$ from true posterior  $p(\cdot | y)$ is minimized among all possible combinations of the given form \eqref{eq:stackForm}.


\paragraph{Hybrid stacking.}
The five stacking methods developed in Sec.~\ref{sec:method} cannot exhaust all plausibility. In particular, given a combination operators $q_{\w}^*(\cdot | y)$, if multiple scores  $U_1, \dots, U_m$ satisfy the proper condition \eqref{eq:proper}, then the augmented score  $U_1 + \cdots  + U_m$ is still valid, and hence existing stacking methods are building blocks toward other stacking approaches.  
For example, when using the density 
mixtures, $q_{\w}^*(\cdot | y) = w_k q_k(\theta|y)$, to maximize the hybrid objective
$ 
{\sum_{n} \log  \sum_{k} w_k q_k(\theta_n| y_n)  }
-\lambda_2
{ \sum_{n} ||\sum_{k} w_k \mu_{kn}- \theta_n  ||^2 }
- \lambda_3
{  (\frac{1}{N}\sum_{n}\sum_{k} {\log(w_k r_{kn})} +1)^2}
$ 
combines needs for KL-closeness, unbiasedness, and rank calibration.

\begin{figure}
    \centering
    \vspace{-0.9em}
\includegraphics[width=\linewidth]{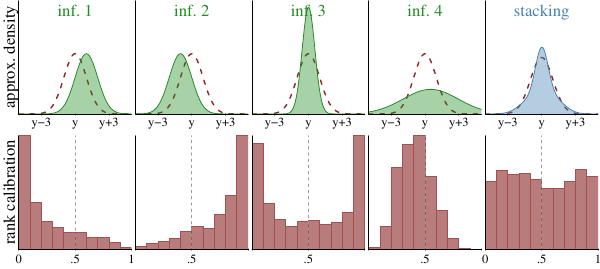}\vspace{-0.6em}
    \caption{Tension among objectives:     Four approximate inferences have the same KL divergence to the true posterior, but differ a lot in the bias, coverage, and rank calibration.}
    \label{fig:ill}
    \vspace*{-1em}
\end{figure}

\setlength{\columnsep}{10pt}%
\begin{wrapfigure}{l}{0.5\linewidth}
 \vspace{-1.7em}
    \includegraphics[width=\linewidth]{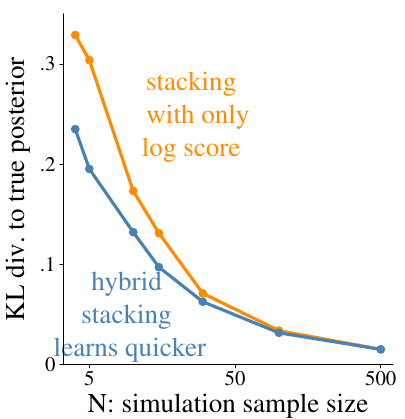}
    \vspace{-1.5em}
    \caption{The KL divergence between 
   true posterior to the stacked posterior  as training size $N$ varies.}  \label{fig:eff}
    \vspace{-1em}
\end{wrapfigure}
Probabilistic distributions on $\R^d$ are infinite dimensional objects. In contrast to all $L^p$ norms that are equivalent in a finite-dimensional Euclidean space, here these scoring rules gauge different projections of the distribution and can provide nearly orthogonal signals. For instance, suppose the true posterior is  $\theta | y \sim \n(y,1)$, the green curves in Fig.~\ref{fig:ill} represent four approximate inferences indistinguishable under the log score as they have the same KL divergence to the true posterior. However, their bias varies from 1 to -1, the real coverage of their nominal 95\% confidence varies from 73\% to 100\%,  and their rank distribution can display severe under- or over-confidence.  Hybrid stacking, shown as the blue curve, makes use of all signals and improves both density-fitting and calibration. Indeed, even when the KL divergence of the posterior is of interest, adding more information such as rank calibration into stacking objectives boosts  efficiency. Fig.~\ref{fig:eff} shows that hybrid stacking has a quicker learning rate, and its posterior inference is more accurate than the plain log score stacking \eqref{eq:mixture_obj} when the training  size $N$ is small. Details of the example are in Appx.~\ref{app:exp}.


\paragraph{General recommendations.}
\emph{Training-validation split:} To avoid overfitting in stacking, we split the simulation table $\{(\theta, y, \tilde \theta)\}$ into training and validation parts. We train individual inferences using the training data and train the stacking weights \eqref{eq_stack} on the validation data. We use extra holdout test data to evaluate the final stacked posterior quality. 
\emph{Fast optimization:} 
All objective functions we derived in Section \ref{sec:method} are (almost everywhere) smooth and straightforward to deploy any (stochastic) gradient optimization recipe. The weights in mixture stacking needs a simplex constraint, for which the multiplicative gradient optimization~\citep{zhao2023generalized} is suitable.
Appx.\ref{app:smooth} discusses smooth approximation of indicator functions.
\emph{Quasi Monte Carlo sampling:} 
When sampling from the stacked inference $\sum_k w_k q_k(\cdot | y)$, the quasi Monte Carlo strategy reduces the variance (Appx.\ref{app:sample}).


\section{Experiments}\label{sec:exp}


\begin{table*}
 \footnotesize
\centering
\begin{tabular}{ c | c c c || c c c || c c c || c c c }
\multicolumn{1}{c|}{\multirow{3}{*}{Task}} &  \multicolumn{3}{c||}{\multirow{2}{*}{Settings}} & \multicolumn{3}{c||}{\multirow{2}{*}{\parbox{3.5cm}{\centering Mixture for KL [Log~Pred. Density] $\uparrow$}}} & \multicolumn{3}{c||}{\multirow{2}{*}{\parbox{3.2cm}{\centering Interval Stacking [Coverage~Error \%]~$\downarrow$}}} & \multicolumn{3}{c}{\multirow{2}{*}{\parbox{3cm}{\centering Moments Stacking [Moments Error] $\downarrow$}}}\\
&&&&&&&&&&&& \\
\cline{2-13}
& dim$(\ve{\theta})$ & dim$(\ve{y})$ & $N$ & Best & Unif. & Stacked & Best & Unif. & Stacked  & Best & Unif. & Stacked\\
\hline
Two Moons & 2 & 2 & 10k & 2.84 & 2.07 & \textbf{2.88} & 3.25 & 5.25 & \textbf{2.75} & -1.72 & -1.41 & \textbf{-1.73} \\
SLCP & 5 & 8 & 10k & -5.60 & -6.15 & \textbf{-5.24} & 3.76 & 5.18 & \textbf{1.16} & 1.03 & 1.25 & \textbf{0.93} \\
SIR & 2 & 10 & 10k & 7.57 & 7.01 & \textbf{7.65} & 2.90 & 9.35 & \textbf{1.55} & -8.86 & -5.19 & \textbf{-8.92} \\ 
SimBIG & 14 & 3,677 & 18k & 15.6 & 16.7 & \textbf{17.0} & 6.08 & 5.85 & \textbf{4.26} & -3.64 & -3.71 & \textbf{-3.79} \\ 
\end{tabular}
\caption{Settings for each task and results for best inferred posterior, uniform ensemble, and stacked posterior. Settings: dimensions of parameters $\ve{\theta}$, data $\ve{y}$, and number of training examples $N$. Mixture for KL: log predictive density on holdout test set  (higher is better). Interval stacking: parameter-averaged coverage error for the 90\% credible intervals on holdout test set (lower is better). Moments stacking: error of posterior moments  \eqref{eq:moment} on holdout test set (lower is better).}
\label{table:exp_all}
        \vspace{-1em}
\end{table*}


\begin{figure}
\includegraphics[width=\hsize]{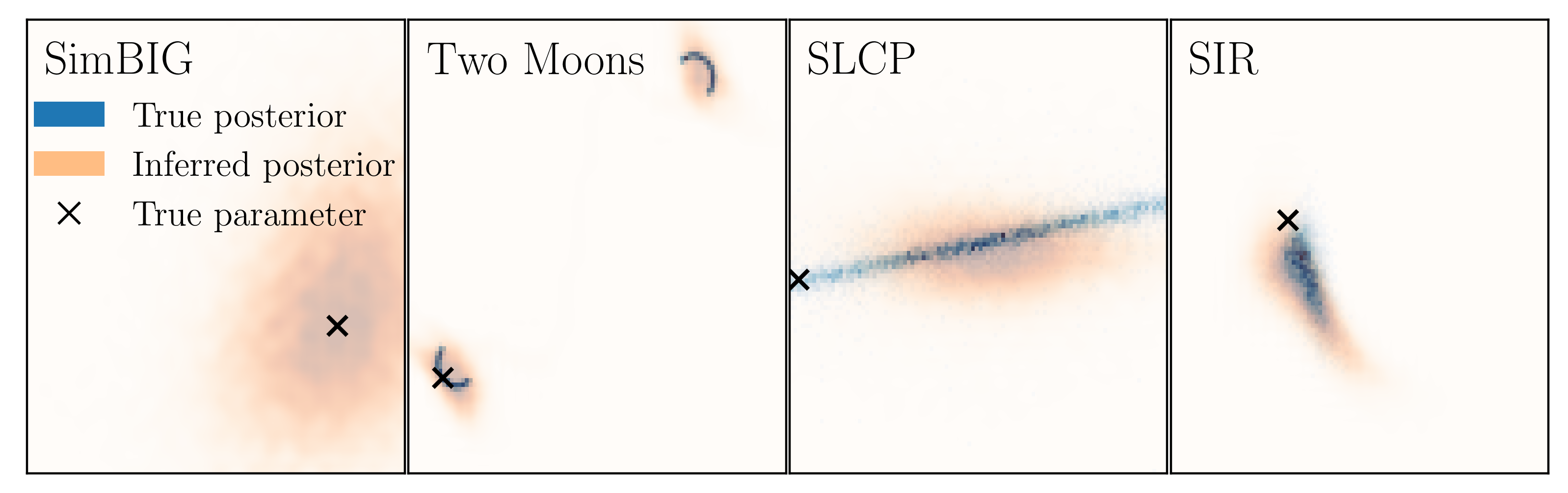}
\vspace{-1em}
\caption{Examples of true and one inferred posteriors in four models. We visualize two margins of the parameter.}\label{fig:posteriors_examples}
\vspace{-1em}
\end{figure}

We conduct stacking experiments on a variety of inference tasks taken from the SBI benchmark \citep[\texttt{sbibm},][]{lueckmann2021benchmarking} and a practical cosmology problem. These tasks are selected to showcase different computational challenges relating to the geometrical complexity of the posterior, the high dimensionality of the parameters/data, or the limited amount of training examples. Table~\ref{table:exp_all} reports the dimensions and number of training examples involved for each task, and Fig.~\ref{fig:posteriors_examples} visualizes posterior examples.

\emph{SBI benchmark}: We consider the \textbf{Two Moons} \citep{greenberg2019automatic}, the simple likelihood and complex posterior~\citep[\textbf{SLCP},][]{papamakarios2019}, and an ODE-based  \textbf{SIR} model.
\emph{Practical cosmology model:} 
We consider a cosmological inference task pertaining to the analysis of galaxy clustering: \textbf{SimBIG}~\citep{hahn2022simbig, regaldo2023}. The SimBIG model  involves 14 key physical parameters to describe the evolution of the Universe. We aim to infer from a vector of 3,677 statistical measurements derived from a large galaxy survey.

For each task, we run $K=50$ (\texttt{sbibm}) / $K = 100$ (SimBIG) independent amortized posterior inferences using the Python package \texttt{sbi}~\citep{tejero2020}. We focus on neural posterior estimators made of conditional normalizing flows. These build on  a masked autoregressive flow \citep[\textbf{MAP}, ][]{papamakarios2017}  architecture, and a multilayer perceptron (MLP) conditioner. Training consists of maximizing the log predictive density using {\sc Adam}~\citep{Kingma2015} over a fixed number of epochs for the \texttt{sbibm} tasks or using an early-stopping procedure for SimBIG: run until the validation loss stops increasing over 20 consecutive epochs. 
For each inference, we randomly select the number of MAF autoregressive layers, number of MLP hidden layers and units, MLP dropout rate, learning rate, and batch size.

In Appx.~\ref{app:exp}, we include additional experiments that runs stacking on neural spline flow \citep[\textbf{NLP}, ][]{durkan2019neural} and other SBI benchmark tasks, and we see the same conclusion holds. We give experiment details in Appx.~\ref{app:exp}. We also provide our Python/PyTorch code on \href{https://github.com/bregaldo/simulation\_based\_stacking}{GitHub} \footnote{\noindent{ {https://github.com/bregaldo/simulation\_based\_stacking.}}}.




\paragraph{Stacking reduces KL gap.} For each task with $K$ inferences, we run  mixture density stacking \eqref{eq:simplex} to 
maximize the log score~\eqref{eq:mixture_obj} trained on  a validation set of 1,000 simulations. We compare in Table~\ref{table:exp_all} the  expected log predictive density of the posterior approximation computed on a holdout set. This is the negative KL divergence from the true posterior up to a constant. We evaluate three meta posterior approximations: (a)~the best single approximation, (b)~a uniform weighting of all approximations, and (c)~stacking. Stacking has the biggest expected log predictive densities in all cases, indicating a closer posterior inference. The same conclusion holds when we run stacking on neural spline flows in the Appx.~\ref{app:exp}.

\begin{figure}
    \centering
    \includegraphics[width=0.48\hsize]{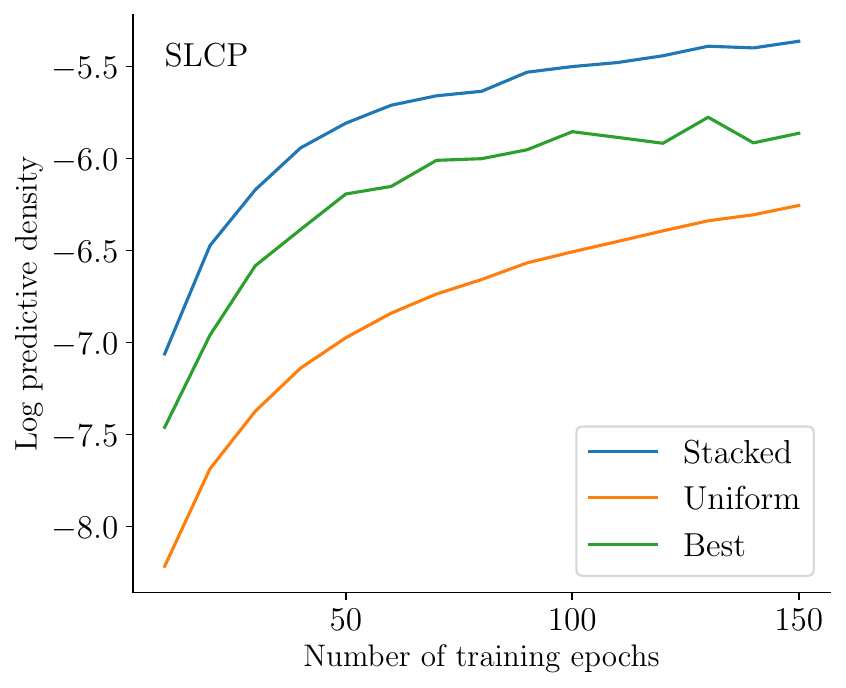}
    ~
      \includegraphics[width=0.47\hsize]{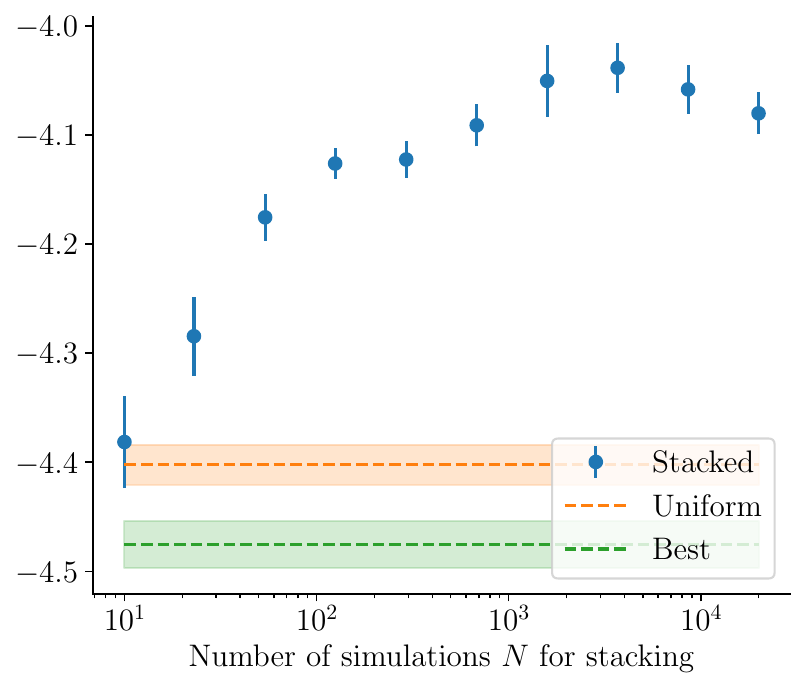}
    \vspace{-0.5em}
    \caption{Mixture for KL: (Left) expected log predictive density as a function of training epochs for the SLCP task. Stacking inferences individually trained over $\sim$~50 epochs performs better than the best single inference  trained over 150 epochs.
    (Right)   evolution of the expected log predictive density with varying stacking simulation size $N$ on the SLCP task.  }
    \label{fig:slcp_log_prob_hist}
        \vspace{-0.7em}
\end{figure}





\paragraph{Stacking performs better with less computation.} 
The stacked posterior $q_{\w}^{\mix}$ of $K$  inferences each trained for a fixed number of epochs can reach a better approximation  than the best single approximation among a series of $K$ inferences trained after a larger number of training epochs. 
In the left panel of Fig.~\ref{fig:slcp_log_prob_hist}, we show  an illustration of this phenomenon for the SLCP task. In 50 epochs, the stacked posterior already performs better than the best single approximation obtained after 150 epochs, which illustrates the interest in stacking for a limited wall time budget. The right panel of Fig.~\ref{fig:slcp_log_prob_hist}  shows how the performance of stacking changes as the  simulation size $N$ varies. We plot the hold-out-data log predictive density of the stacked posterior with neural spline flows, constructed from varying simulation size $N$. With a very small sample size, stacking can suffer from overfitting and regularization can help.

\begin{figure}
\includegraphics[width=\hsize]{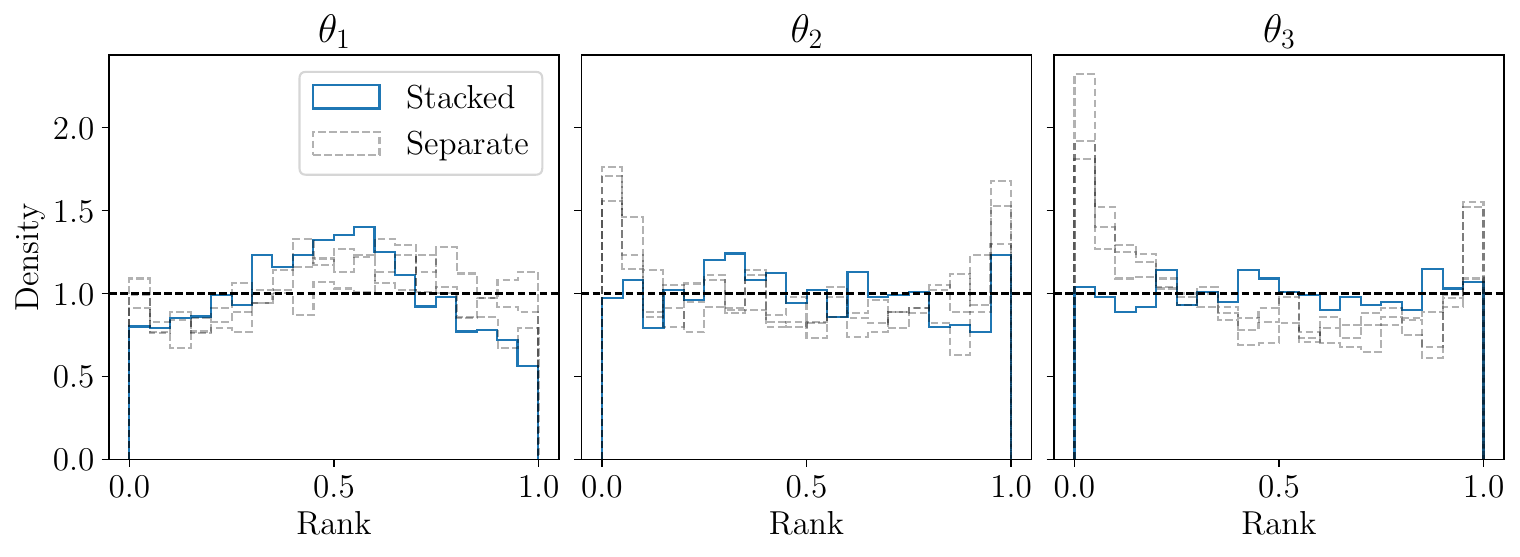}
        \vspace{-0.5em}
\caption{Ranks statistics before and after rank stacking for the SimBIG task. Stacked rank statistics are closer to the uniform distribution (black dashed lines), thus indicating better calibration than individual posteriors.}\label{fig:rank_simbig}
        \vspace{-1em}
\end{figure}

\paragraph{Stacking calibrates rank statistics.} We run stacking for rank calibration \eqref{eq:rank_obj} on the SimBIG task.  Fig.~\ref{fig:rank_simbig} shows the rank statistics obtained from the optimal stacked posterior for the first three cosmological parameters, and compares them to the same rank statistics obtained from three individual posteriors. 
We constrain the rank statistics of each dimension simultaneously. There is a clear improvement on ranks of parameters $\theta_2$ and $\theta_3$, approaching uniformity. However, for parameter $\theta_1$ the stacked rank statistics display the same kind of underconfidence patterns as individual posteriors. It happens that all individual neural posterior approximations are underconfident for this dimension $\theta_1$.  Because any mixture of underconfident approximations remains  underconfident, rank-stacking cannot help on the $\theta_1$ margin.

\paragraph{Stacking calibrates confidence intervals.} For each task, we perform interval stacking as described in Sec.~\ref{sec:interval} for all parameters simultaneously and focusing on $90\%$ central confidence intervals (i.e. ${\alpha = 0.1}$). For any scalar parameter $\theta$, we compute the coverage under the true posterior on holdout test data ${C_\alpha = \mathbb{E}_{p(y, \theta)}\left[\mathds{1}\{\theta \in I_{\alpha, q(\cdot| y)}\}\right]}$ where $I_{\alpha, q(\cdot | y)}$ is the central confidence interval in  approximation $q(\cdot | y)$. If the approximation is perfect, then $C_\alpha$ should be $1 - \alpha$. Table~\ref{table:exp_all} reports coverage error $\lvert C_\alpha - (1-\alpha) \rvert$, averaged over all parameters, for best single approximation, uniform ensemble, and stacking. In every task, interval stacking clearly improves coverage.

\paragraph{Stacking calibrates moments.} We run moment stacking to calibrate the posterior means and variances by optimizing the moments objective \eqref{eq:moment}. We compare in Table~\ref{table:exp_all} the expected error \eqref{eq:moment} of posterior means and variances of best single approximation, uniform ensemble, and  moment stacking. All errors are computed on holdout test set. Our moment stacking methods outperforms for all tasks.


\section{Discussions}\label{sec:discuss}
\textbf{Stacking/model averaging.}
Stacking \citep{wolpert1992stacked,  breiman1996stacked} is an old idea to combine learning algorithms. Classic stacking combines point predictions only. Recently stacking is advocated to combine Bayesian outcome-predictive-distributions \citep{clyde2013bayesian,  le2017bayes, yao2018using}, since it is more robust against model misspecification than Bayesian model averaging \citep[e.g.,][]{hoeting1999bayesian}.  Traditional stacking aims at the prediction of data $y$ and uses the exchangeability therein. Suppose $\{y_{i}\}_{i=1}^N$ are IID observations, and  $p_{k}(y_i)$ is the predictive density of $y_i$ in the $k$-th model, which is only available in likelihood-tractable settings, then stacking seeks to maximize 
$\sum_{i=1}^n (\log \sum_{k=1}^K w_k p_{k}(y_i)),
$  
so that the combined predictive density is close to the true data-generating process. In contrast, our paper aims at Bayesian computation and uses the exchangeability of amortized simulations. We do not need a tractable likelihood, nor any structure of data $y$.
Although scoring rules is not a new idea to Bayesian model \emph{evaluation}  \citep{gneiting2007strictly, vehtari2012survey}, as far as we know, traditional Bayesian stacking is not beyond the log scores and mixture until this paper, while the present paper introduces new combination forms and distributional scoring rules to \emph{learn} stacking weights. Our stacking approach is useful not only to combine different posterior approximations in Bayesian computation, but also to combine different probabilities models to fit observations. In a companion paper \citep{yao2024stacking}, we establish a general stacking framework to combine probabilistic predictive distributions and  provide more theory discussions.  

\paragraph{Meta-learning for multi-run Bayesian computation.}
Modern hardware have attracted the development of parallel Bayesian inferences. One strategy is to tailor MCMC tuning criterion for  parallel runs to boost mixing \citep{hoffman2021adaptive}. Another strategy is to run inference methods on subsamples of the dataset and combine the subsampled posteriors to be an unnormalized product $\prod_k q_k (\theta|y)$ \citep[e.g.,][]{nemeth2018merging, mesquita2020embarrassingly, agrawal2021amortized}. More generally, it is appealing to run many shorter, and potentially biased inferences and combine them. In this direction, the most related approach to our paper is to use stacking in non-mixing Bayesian computations \citep{yao2022stacking}. Despite a similar title, the existing stacking-for-computation approach aims to improve how good the statistical model predicts future outcomes. It mixes posterior predictive distributions to optimize the data score, $\E \log p(\tilde y) =\int_{\Theta} p(\tilde y |\theta) q(\theta | y) d\theta$, which is only tractable with a known likelihood, and arguably less relevant to scientific computing where parameters have physical meanings. Our paper has a fundamentally different goal on the inference accuracy. 

\paragraph{Simulation-based inference and calibration.}
Many individual objective functions of our stacking have been used as part of simulation-based inference or calibration. Maximizing the stacked log predictive density shares the same goal of minimizing $\KL(p, q)$ as in the traditional neural posteriors. 
Simulation-based calibration \citep{cook2006validation, talts2018validating, modrak2023simulation} has examined the marginal rank statistics for testing, while we use it for training. 
Under the repeated prior sampling
and correct computation,  Bayesian models are calibrated \citep{dawid1982well}. Our computation calibration should not be confused with the frequentist calibration \citep[repeated data sampling under one true parameter,][]{masserano2022simulation}.
The sample-based stacking is relevant to the discriminative calibration  \citep{yao2023discriminative} and the adversarial training \citep{ramesh2022gatsbi}. The moment matching shares a similar objective with the moment calibration \citep{yu2021assessment} and moment network \citep{jeffrey2020solving}, while our new objective combines two moments.
Our paper differs from these existing tools in that we combine multiple inferences.  



 \paragraph{Limitations and future directions}
This paper develops a stacking strategy to combine multiple simulation-based inferences for the same task. We design stacking to incorporate various combination forms and utility functions for distributional approximation.

Our stacking utility  function is averaged over $y$, which computes the averaged approximation quality. We have discussed the possibility of local weights, but more evaluation is needed in the future.

Including stacking in the inference pipeline provides double robustness. If individual inferences are accurate enough, there is no need for stacking; If the posterior stacking model is flexible enough, individual inferences can be arbitrarily off.  
In the experiments we tested the individual inferences are well-constructed, while the stacking model is relatively simple with a relatively negotiable computation cost. Looking forward, with advances in multiple-data processors such as GPUs, it is faster to run a large number of crude approximations in parallel than to optimize one single run to full precision, making it appealing to use a comprehensive stacking model to combine many cheaper inferences, which we leave for future work.

\subsubsection*{Acknowledgements}
The authors thank Andrew Gelman and Andreas Buja for helpful discussions.

\small 
\bibliographystyle{apalike}
\bibliography{main}	
\normalsize

\newpage
\appendix
\normalfont
\onecolumn

\fancyhead[C]{Simulation Based Stacking, Appendix}
\fancyhead[R]{\thepage}

\renewcommand{\thesection}{\Alph{section}}

\section*{\centering \NoCaseChange{Appendices to ``Simulation-Based Stacking''}}  

\section{\NoCaseChange{Background}}\label{app:back}
\paragraph{Simulation-based inference (SBI).}
In a typical Bayesian inference setting, we are interested in the posterior inference of parameter $\theta \in \R_d$ given observed data $y$.  The ultimate goal is to infer a version of the condition distribution $p(\theta|y)$ and/or sample from it. 
In simulation-based inference, we cannot evaluate the likelihood, but we can easily simulate outcomes from the data model $y|\theta$. We create a  size-$N$ joint simulation table  $\{(\theta_n, y_n)\}_{n=1}^N$ by first drawing parameters $\theta_n$ from the prior distribution $p(\theta)$ and one realization of data $y_n$ from the data-forward model $p(y|\theta_n)$.

\paragraph{Normalizing flow.}
The posterior inference task becomes a conditional density estimation task. In one neural posterior estimation, we parameterize the posterior density as a normalizing flow $q_{\beta}(\theta|y)$,  where $\beta \in \R_m$ is the normalizing flow parameter. More concretely, let $z$ be a multivariate standard Gaussian random variable in $\R_d$, consider $\theta=f_{\beta, y}(z)$, a bijective mapping from $z \in \R_d$ to $\theta \in \R_d$. Let $z=F_{\beta, y}(\theta)$ be the inverse of this mapping. From change-of-variable, the implied distribution of $\theta$ is $q_{\beta}(\theta |y) = p_{\mathrm{MVN}}(z=F_{\beta, y}(\theta)) | \mathrm{det} ( \frac{d}{d\theta} F_{\beta, y}(\theta))|$, where $p_{\mathrm{MVN}}$ denotes the density of standard Gaussian. When the bijective $f_{\beta, y}$ is flexible enough, in principle, the family of derived densities $\{q_{\beta}(\theta|y): \beta\in \R_m\}$ covers all smooth conditional distributions on $\R_d$.

\paragraph{Neural posterior estimation.}
Given any $\beta$, $q_{\beta}(\theta|y)$ is ensured to be a normalized conditional density on $\Theta$  by design: $q_{\beta}(\theta|y)\geq 0$ and 
$\int_{\Theta} q_{\beta}(\theta|y) d\theta =1$. Since we have simulations from the joint density $p(\theta, y)$, we fit this normalizing flow to minimize the KL divergence: 
$$
\hat \beta =  \underset{\beta \in \R_m}{\mathrm{arg~max}} \sum_{i=1}^N \log q_{\beta}(\theta_i|y_i).
$$
This inferred $q_{\hat \beta}(\theta|y)$ is one neural posterior estimation. We are able to (a) evaluate the density $q_{\hat \beta}(\theta|y)$ for any $(\theta, y)$ pair, and (b) given any $y$, sample IID draws $\tilde \theta_1, \dots,  \tilde \theta_S$ from  $q_{\hat \beta}(\theta|y)$. 

By varying the normalizing flow architecture or hyperparameters, we obtain multiple neural posterior estimations $\{q_{1}(\theta|y), q_{2}(\theta|y), \dots,  q_{K}(\theta|y)\}$.
This present paper aims to aggregate them to provide better inference.

\section{\NoCaseChange{Additional Theory and Proof}}\label{app:theory}

\paragraph{Recap of the general SBI stacking setting.}
Consider a parameter space $\Theta = \R_d$ with the usual Borel measure, and any measurable data space  $\mathcal{Y}$. 
We are given a sample of $N$ IID draws $\{(\theta_i, y_i): 1\leq i \leq N\}$ from a joint distribution $p(\theta, y)$ on the product space $\Theta\times \mathcal{Y}$. Let $p(\theta|y)$ be a version of the true conditional density. We are given $K\geq 2$ conditional densities $q_k(\theta|y)$. Besides, for each inference index $k\leq K$ and simulation index $n\leq K$, we have obtained an IID sample of $S$ draws: $\tilde \theta_{kn1}, \dots, \tilde \theta_{knS} \sim q_k( \theta|y_n)$. We always denote $k$ the index of inference, $n$ the index of simulations, and $s$ the index of the posterior draw. See Figure \ref{fig:ill} for visualization.

Typically we have a training-validation spilt to avoid over-fitting. For the theory part, it suffices to assume that the $K$ conditional densities are fixed, and will not change as $N$ increases.

\subsection{General posterior stacking (Proposition \ref{thm:consi})}
We give a formal definition of Proposition \ref{thm:consi} in the paper.

In the general posterior stacking, we specify a combination form, and an objective function of distributional approximation. Suppose we can evaluate the score of the stacked distribution, 
$U(q_{\w}^*(\cdot | y_n), \theta_n)$. Further if the scoring rule $U$ is proper \eqref{eq:proper}.  We combine $K$ posterior inferences with the form \eqref{eq:stackForm}, and fit stacking parameters $\w$ via a sample M-estimate, 
\begin{equation}
 \hat \w = \arg \max \sum_{n=1}^N U(q_{\w}^*(\cdot | y_n), \theta_n). 
\end{equation}
The expectation 
$\E_{p(y, \theta)} U(q_{\hat \w}^*(\cdot | y), \theta)$  is the average utility function of the stacked posterior.

\setcounter{proposition}{0}
\begin{proposition}  
Clause I (population utility). For any given $\w$, 
as $N\to \infty$, 
$$\frac{1}{N} \sum_{n=1}^N U(q_{\w}^*(\cdot | y_n), \theta_n) = \E_{p(y, \theta)} U(q_{\w}^*(\cdot | y), \theta)+ o_p(1).
$$

Clause II (asymptotic optimality).
For any $\epsilon>0$, any given $\w^{\prime}$, 
as $N\to \infty$,  
$$\Pr \left( \E_{p(y, \theta)} U(q_{\hat \w}^*(\cdot | y), \theta) \leq    \E_{p(y, \theta)} U(q_{\hat \w^{\prime}}^*(\cdot | y), \theta)+\epsilon \right) \to 1.$$

Clause III (convergence). Further assume that (a) the support of $\w$ is compact, (b) there is a true $\w_0$, such that $q_{ \w_0}^*(\cdot | y) = p(\theta|y)$ almost everywhere,  (c) the combination is locally identifiable at the truth, i.e., if $q_{\hat \w}^*(\cdot | y) = p(\theta|y)$ then $\w = \w_0$, and (d) the scoring rule $U$ is strictly proper, then the stacking weight estimate is consistent as $N\to \infty$,
$$\hat w = w_o + o_p(1).$$
\end{proposition}
\begin{proof}
We briefly sketch the proof since similar proofs have appeared in previous propositions.

Clause I is from the weak law of large numbers.

Clause II addresses the optimality of the resulting divergence metric rather than the estimated $\hat \w$. From WLLN, 
for any $\epsilon>0$, 
as $N\to \infty$,  
$$\Pr \left( \left| \frac{1}{N}\sum_{n=1}^N U(q_{\hat w}^*(\cdot | y_n), \theta_n) - \E_{p(y, \theta)} U(q_{\hat \w}^*(\cdot | y), \theta) \right|\geq
1/2 \epsilon \right) = o(1).$$
$$\Pr \left( \left| \frac{1}{N} \sum_{n=1}^N U(q_{\w^\prime}^*(\cdot |  y_n), \theta_n) - \E_{p(y, \theta)} U(q_{\w^\prime}^*(\cdot | y), \theta) \right|\geq
1/2 \epsilon \right)  = o(1).$$
Furthermore, by definition,  
$$\frac{1}{N} \sum_{n=1}^N U(q_{\w^\prime}^*(\cdot |  y_n), \theta_n) \leq \frac{1}{N} \sum_{n=1}^N U(q_{\hat \w}^*(\cdot |  y_n), \theta_n).$$
Combine these three lines, we have
\begin{align*}
&\Pr \left( \E_{p(y, \theta)} U(q_{\hat \w}^*(\cdot | y), \theta) \leq    \E_{p(y, \theta)} U(q_{\hat \w^{\prime}}^*(\cdot | y), \theta)+\epsilon \right)\\ 
&= \Pr \left( \E_{p(y, \theta)} U(q_{\hat \w}^*(\cdot | y), \theta) \leq    \E_{p(y, \theta)} U(q_{\hat \w^{\prime}}^*(\cdot | y), \theta)+\epsilon,   ~~ \frac{1}{N} \sum_{n=1}^N U(q_{\w^\prime}^*(\cdot |  y_n), \theta_n) \leq \frac{1}{N} \sum_{n=1}^N U(q_{\hat \w}^*(\cdot |  y_n), \theta_n)\right) \\
&\leq \Pr \left( \left| \frac{1}{N}\sum_{n=1}^N U(q_{\hat w}^*(\cdot | y_n), \theta_n) - \E_{p(y, \theta)} U(q_{\hat \w}^*(\cdot | y), \theta) \right| 
+ \left| \frac{1}{N} \sum_{n=1}^N U(q_{\w^\prime}^*(\cdot |  y_n), \theta_n) - \E_{p(y, \theta)} U(q_{\w^\prime}^*(\cdot | y), \theta) \right| \geq \epsilon\right)\\
&\leq \Pr \left( \left| \frac{1}{N}\sum_{n=1}^N U(q_{\hat w}^*(\cdot | y_n), \theta_n) - \E_{p(y, \theta)} U(q_{\hat \w}^*(\cdot | y), \theta) \right|\geq
1/2 \epsilon \right) \\
&+
\Pr \left( \left| \frac{1}{N} \sum_{n=1}^N U(q_{\w^\prime}^*(\cdot |  y_n), \theta_n) - \E_{p(y, \theta)} U(q_{\w^\prime}^*(\cdot | y), \theta) \right|\geq
1/2 \epsilon \right)\\
&=o(1),
\end{align*}
which proves Clause II.

Clause III is a direct application of Lemma \ref{lemma:M}. Here we assume a compact $\w$ space to ensure uniform convergence.
\end{proof}

\subsection{Convergence in mixture stacking (Prop.~\ref{thm:KL})}

First, we consider mixture stacking with the log density objective. Given a simplex weight $\w\in \s_K$, The mixed posterior density \eqref{eq:simplex} is  $q^{\mix}_{\w}(\theta|y) \coloneqq \sum_{k=1}^K w_k q_k(\theta|y)$. It is straightforward to derive the well-known correspondence between the log density and the Kullback–Leibler (KL) divergence.

\begin{proposition}\label{thm:KL}
Clause (I). For any $\w\in \s_K$, as $N\to \infty$, the stacking objective converges in distribution to the negative conditional KL divergence between the true posterior and stacked approximation, i.e.,  
$$ \frac{1}{N} \sum_{i=1}^N \log q^{\mix}_{\w}(\theta_i|y_i) + \KL (p(\theta|y), q^{\mix}_{\w}(\theta|y)) -C = o_{p}(1),
$$
where $C= \log p(\theta|y) p(\theta, y) $ is a constant that does not depends on $\w$ or $q$.

Clause (II). If (a) there exists a $\w_0\in \s^K$, such that $p(\theta|y) =  q^{\mix}_{\w_0}(\theta|y)$ almost everywhere, and (b) the mixture form \eqref{eq:simplex} is locally identifiable at the truth, i.e., if there is there is a $\w^\prime$ such that  $q^{\mix}_{\w^\prime}(\theta|y)=q^{\mix}_{\w_0}(\theta|y)$ almost everywhere, then  $\w^\prime=\w_0$.
Then the optimal stacking weight $$\hat \w = \arg \max_{\w\in \s_K} \frac{1}{N} \sum_{i=1}^N \log q^{\mix}_{\w}(\theta_i | y_i) $$ converges to the true 
$\w_0$ in probability, i.e., for any $\epsilon>0$,
$$ \Pr_{p}(|\hat \w - \w_0 |\geq \epsilon) \to 0, ~~\mathrm{as}~ N\to \infty.$$
\end{proposition}

\begin{proof}
    Clause (I) is a direct application of the weak law of large numbers.   Since $\{(\theta_i,y_i)\}$ are IID draws from the joint,  in probability we have 
    \begin{align*}
  \frac{1}{N} \sum_{i=1}^N \log q^{\mix}_{\w}(\theta_i|y_i) &\stackrel{p}{\to} \int_{\Theta\times\mathcal{Y}} \log q^{\mix}_{\w}(\theta|y) p(\theta, y)d\theta dy\\
   &= \int_{\Theta\times\mathcal{Y}} q^{\mix}_{\w}(\theta|y) p(\theta| y) p(y)d\theta dy\\
   &= -\int_{\Theta\times\mathcal{Y}} \left(\log p^{\mix}_{\w}(\theta|y) - \log q^{\mix}_{\w}(\theta|y) \right) p(\theta| y) p(y)d\theta dy + \int  \log p^{\mix}_{\w}(\theta|y)  p(\theta, y) d\theta dy\\
   &=-\KL (p(\theta|y), q^{\mix}_{\w}(\theta|y)) +C.
    \end{align*}   
\end{proof}

Clause (II) is a consequence of the convergence of the maximum likelihood estimation (MLE). In order to apply to other proportions, we state 
the  general M-estimation theory. For example, the following lemma is from \citet{van2000asymptotic}.

\begin{lemma}\label{lemma:M}
Assuming $y_1, \dots y_N$ are IID data from $p_{\beta_0}(y)$.  Let $M_n$ be random functions and let M be a fixed function of the parameter $\beta$ such that for every $\epsilon >0$,
\begin{equation}\label{eq:M_1}
\sup_{\beta} |M_n(\beta)- M(\beta)| \pto  0, 
\end{equation}
\begin{equation}\label{eq:M_2}
\sup_{\beta: d(\beta, \beta_0)\geq \epsilon} M(\beta)< M(\beta_0).    
\end{equation}
Then any sequence of estimators $\hat \beta_n$ with 
$$M_n(\hat \beta) \geq  M_n(\beta_0) - o_p(1)$$
converges in probability to $\beta_0$.
\end{lemma}

In the context of Clause (II), the stacking parameter $\w$ is on a compact space $\s_K$.  Because (a) the identifiable assuming and (b) the unique minimizer to $\E_p\log q(x)$ is $p=q$, the true weight $\w_0$ is the unique minimize of 
$\int_{\Theta\times\mathcal{Y}} q^{\mix}_{\w}(\theta|y) p(\theta| y) p(y)d\theta dy$, i.e., for any $\epsilon>0$, and any $\w^\prime$ such that 
$||\w_0 - \w^{\prime}||>\epsilon$, we have 
$$\int_{\Theta\times\mathcal{Y}} q^{\mix}_{\w_0}(\theta|y) p(\theta| y) p(y)d\theta dy > \int_{\Theta\times\mathcal{Y}} q^{\mix}_{\w^\prime}(\theta|y) p(\theta| y) p(y)d\theta dy,$$
which verifies condition \eqref{eq:M_2}, while the WLLN ensures the uniform convergence condition 
\eqref{eq:M_1}. Applying Lemma \ref{lemma:M} proves Clause II.


\subsection{Rank-based calibration and stacking (Prop.~\ref{thm:rank_div} \& \ref{thm:rank_linear})}

We now deal with rank-based divergence and stacking. We assume $\Theta=\R$ for the next two propositions since we only compute marginal ranks.

\begin{proposition}\label{thm:rank_div}
The rank-based metric $D_{\mathrm{rank}}(p, q)$ defined in \eqref{eq_rankdiv} is a generalized divergence:  (I) $D_{\mathrm{rank}}(p, q) \geq 0$ 
for any $p$ and $q$. (II)  When $q(\theta|y ) = p(\theta|y)$ almost everywhere,  $D_{\mathrm{rank}}(p, q) = 0$. 
\end{proposition}

\begin{proof}
From the definition of $D_{\mathrm{rank}}(p, q)$, 
let $(\theta,y) \sim p(\theta, y)$ be a pair of random variables from $p$, and let $x_1=  F_{p}(\theta|y)$ and $x_2=  F_{q}(\theta|y)$ be two transformed random variables, then  $D_{\mathrm{rank}}(p, q) = D_{0}(x1, x2)$. Because $D_{0}$  is a well-defined divergence,   $D_{\mathrm{rank}}(p, q) = D_{0}(x1, x2) \geq 0$ for any $p$ and $q$.  If $p=q$, then $x_1$ and $x_2$ have the same distribution and hence $D_{\mathrm{rank}}(p, q) = D_{0}(x1, x2)=0$.
\end{proof}

$D_{\mathrm{rank}}(p, q)$ is CDF-based. Its empirical estimates use rank only. It might be unsatisfactory that $D_{\mathrm{rank}}(p, q)$ is not a strict divergence, meaning that there can be a distinct pair of joint distributions $p\neq q$, such that $D_{\mathrm{rank}}(p, q)=0$. Indeed, for any $p(\theta, y)$, let $q(\theta|y)\coloneqq p (\theta)$, then $D_{\mathrm{rank}}(p, q)=0$. This edge case is a well-known example where the traditional rank-based calibration is not sufficient and can lead to false-negative testing. That said, the rank-based divergence has the advantage that it is rank/sample only; no density evaluation of $q$ is needed. We find that the rank-based divergence is particularly powerful when augmented with other density-based divergences.

Let us briefly recap the rank-related definitions.
We still consider the mixture stacking form  \eqref{eq:simplex}: given a simplex weight $\w\in \s_K$, the mixed posterior density is  $q^{\mix}_{\w}(\theta|y) \coloneqq \sum_{k=1}^K w_k q_k(\theta|y)$. We assume all conditional densities  $q_k(\theta|y)$ are continuous on $\Theta$.
Let $r_{kn}$ is the  rank statistics  defined in \eqref{eq:rank}:
$
r_{kn} \coloneqq \frac{1}{S}\sum_{s=1}^S \mathds{1} \{\theta_{n} \geq  \tilde \theta_{kns}\}. 
$

\begin{proposition}\label{thm:rank_linear}
Clause (I). For any fixed value $\mu\in \Theta$, and any weight $\w\in\s_K$,
let $$F_{k,y}(\mu)\coloneqq \Pr_{\theta \sim q_k(\theta|y)}(\theta < \mu), ~~~~\mathrm{and}~~  F_{y}^*(\mu)\coloneqq\Pr_{\theta \sim q^{\mix}_{\w}}(\theta < \mu)$$ be the CDFs in the individual conditional distributions and the mixture, then the CDF remains linear: 
\begin{equation}\label{eq_linear_rank_po}
F_{y}^*(\mu)= \sum_{k=1}^K w_k F_{k,y}(\mu).    
\end{equation}

Clause (II). For any simulation index $n$, define   
$$r^{\mix}_{n} \coloneqq \sum_{k=1}^K w_k r_{kn}$$
be the stacked rank,
then for any weight $\w\in\s_K$, this linearly additive rank converges to the mixture CDF: 
$$
r^{\mix}_{n}  \to \Pr_{\theta \sim q^{\mix}_{\w}(\theta | y_n) }( \theta \leq \theta_{n} ), ~~ \mathrm{almost~~surely,  ~~~as~~} S\to \infty. 
$$
\end{proposition}

\begin{proof}
Clause I is due to the integral linearity.  The CDFs are integrals on a fixed half interval:
$$F_{k,y}(\mu)= \Pr_{\theta \sim q_k(\theta|y)}(\theta < \mu) = \int_{-\infty}^\mu q_k(\theta|y) d\theta, 
$$
Likewise, 
\begin{align*}
F_{y}^*(\mu)
&= \int_{-\infty}^\mu q_k(\theta|y) q^{\mix}_{\w}(\theta | y) d\theta \\   
&=\int_{-\infty}^\mu \left(q_k(\theta|y) \sum_{k=1}^K w_k q_k (\theta | y) \right)d\theta \\    
&= \sum_{k=1}^K w_k \left(\int_{-\infty}^\mu q_k(\theta|y)  d\theta\right) \\  &= \sum_{k=1}^K w_k F_{k,y}(\mu),
\end{align*}
which proves Clause 1.
 
Clause 2 states the convergence of the mixed ranks.
For any fixed $k$ and $n$, 
because $\tilde \theta_{kns}$ are IID draws from $q_k(\theta|y_n)$, the strong law of large numbers applies: 
$$r_{kn}= \frac{1}{S}\sum_{s=1}^S \mathds{1} \{\theta_{n} \geq  \tilde \theta_{kns}\} \to \Pr_{\theta\sim q_k(\theta|y_n) }(\theta\leq \theta_n), ~~\mathrm{almost~~surely.}
$$
By an elementary probability lemma of the sum of almost surely convergent (Lemma \ref{lamma_sum}), the mixed rank converges almost surely as well, 
\begin{align*}
r^{\mix}_{n} = \sum_{k=1}^K w_k r_{kn} &\to \sum_{k=1}^K w_k  \Pr_{\theta\sim q_k(\theta|y_n) }(\theta\leq \theta_n),
~~\mathrm{almost~~surely,}\\
&=\Pr_{\theta \sim q^{\mix}_{\w}(\theta | y_n) }( \theta \leq \theta_{n} ),
\end{align*}
where the last equality is from Equation \eqref{eq_linear_rank_po}.   
\end{proof}

\begin{lemma}\label{lamma_sum}
Let $\{X_n\}$ and $\{Y_n\}$ be two sequences of random variables. If $X_n$ converges almost surely to $X$ and $Y_n$ converges almost surely to $Y$, then the sum $X_n + Y_n$ converges almost surely to $X + Y$.
\end{lemma}

\subsection{Sample-based stacking for JS divergence (Prop.~\ref{thm:JS})}

In sample-stacking, we aggregate individual inferences in \eqref{eq:sample_stack_linear}. $\theta^*_{ns} = \w_0 + \mathbf{w}_1 \tilde   \theta_{1ns} + \dots + \mathbf{w}_K \tilde  \theta_{Kns}$ is the aggregated inference sample. 
For any given $\w$ and any $y$, we define the distribution of the sample-aggregated inference as follows. Let $\tilde \theta_k$ be a random variable draws from $q_k(\theta|y)$, these $\tilde \theta_k$ are mutually independent, and denote $q^*_{\w}(\theta|y)$ be law of the random variable $\w_0 + \mathbf{w}_1 \tilde \theta_{1} + \dots + \mathbf{w}_K \tilde \theta_{Kns}.$

From each simulation draw 
$$(\theta_n, y_n, \theta^*_{n1}, \theta^*_{n2}\dots, \theta^*_{nS}),$$
we generate $(S+1)$ classification examples with feature $\phi$ and label $z$:
$$ z=1, ~\phi=(\theta_n, y_n)$$
$$ z=0, ~\phi=(\theta^*_{n1}, y_n)$$
$$ \cdots$$
$$ z=0, ~\phi=(\theta^*_{nS}, y_n)$$

Let  $f(z=1|\phi) = \Pr(z=1|\phi, \mathrm{some~classifier})$ be any classification probability prediction that uses $\phi$ to predict $z$. Let $\mathcal{F}$  be the space of all such binary classifiers.
To balance the two classes, we typically use a weighted classification utility function 
$$U(z, \phi, f) = \frac{C}{S+1} \mathds{1} (z=1) \log f(z=1|\phi) + \frac{CS}{S+1}  \mathds{1} (t=0) \log f(z=0|\phi),$$
where $C=\frac{(S+1)^2}{2S}$ is a normalizing constant.

Sample stacking solves a mini-max optimization. 
\begin{equation} \label{eq_weighed_ad}
    \hat \w =  \arg\min_{\mathbf{w}}  \max_{f\in \mathcal{F}} \sum_{i=1}^{N(S+1)}U(z_i, \phi_i, f).    
\end{equation}

\begin{proposition}\label{thm:JS}
Clause (I). For any fixed $\w$ and fixed $S\geq 1$, as $N\to \infty$, the best classifier utility corresponds to the conditional  Jensen Shannon divergence between the aggregated inference and the truth,
$$
\max_{f\in \mathcal{F}} \frac{1}{N}\sum_{i=1}^{N(S+1)}U(z_i, \phi_i, f)
=
\E_{y}\left[ \frac{1}{2}p(\theta|y)
  \log \frac{p(\theta|y)}{r(\theta|y)} + \frac{1}{2}q^*_{\w}(\theta|y)
  \log \frac{p(\theta|y)}{r(\theta|y)} 
 \right] - \log2 +o_p(1),  
$$
where $$r(\theta|y)\coloneqq \frac{1}{2}\left(p(\theta|y)+q^*_{\w}(\theta|y)\right).$$

Clause (II). 
Let $\hat w$ be the solution from sample stacking \eqref{eq_weighed_ad}.  Assuming 
(a) there exists one true $\w_0$ such that $q^*_{\w}(\theta|y)= p(\theta|y),$ and (b) the sample model is locally identifiable at the truth, i.e., if $q^*_{\w^{\prime}}(\theta|y)= p(\theta|y)$, then  $\w^{\prime}=\w_0$. 
Then for any fixed $S$, as $N\to \infty$, the stacking weight estimate is consistent in probability,
$$\hat \w = \w_0 +o_p(1).$$
\end{proposition}
Clause (I) is similar to the standard adversarial learning, \citep[e.g., Theorem 8 in][]{yao2023discriminative}. Clause (II) can be proved from Lemma \ref{lemma:M}. It is a direct consequence of the main proposition.

\subsection{Interval stacking (Prop.~\ref{thm:interval})}
In interval stacking, we run stacking to solve 
\begin{align*}
   \min_{\w}  & \sum_{i=1}^N U((r^*_n,l^*_n), \theta_n ), ~~~~
U((r^*_n,l^*_n), \theta_n ) \coloneqq 
(r^*_n - l^*_n)  \\ &
+ \frac{2}{\alpha} (l^*_n - \theta_n) \mathds{1} (\theta_n< l^*_n) 
+ \frac{2}{\alpha} (\theta_n - r^*_n) \mathds{1} (\theta_n> r^*_n).  
\end{align*}

\begin{proposition}\label{thm:interval}
Clause (I). For any fixed confidence-level $\alpha\in (0,1)$,  and for any $y$, let $r_{k, y}$ and $l_{k, y}$ be the $\frac{\alpha}{2}$ and $(1-\frac{\alpha}{2} )$ quantile of the distribution $q_k(\theta|y)$. Given a weight $\w$, $l^*_{\w, y} = \sum_{k=1}^K w_{k} l_{k, y}$ and $r^*_{\w, y} = \sum_{k=1}^K \w_{k+K} r_{k, y}$  are the stacked left and right confidence interval endpoint. For any fixed $\w$ and as $N\to \infty$,
\begin{equation}\label{eq_interval_pop}
  \frac{1}{N}\sum_{i=1}^N U((r^*_n,l^*_n), \theta_n ) - \E_{p(\theta,y)} U( r^*_{\w, y}, l^*_{\w, y},  \theta ) = o_p(1).  
\end{equation}

Clause (II). Assuming that $p(\theta|y)>0$ for $\theta \in \Theta$ and $y$ almost everywhere.  For any fixed confidence-level $\alpha\in (0,1)$, let $l_p(y), r_p(y)$ be the $\frac{\alpha}{2}$ and $(1-\frac{\alpha}{2} )$ quantile of the true posterior distribution $p(\theta|y)$, they are the unique minimize to the population limit in \eqref{eq_interval_pop}. That is, for any two  functions  $l(y), r(y)$ that maps $y$ to the $\Theta$ space,
$$\E_{p(\theta,y)} U( r_p(y), l_p(y),  \theta ) \leq 
\E_{p(\theta,y)} U( r(y), l(y),  \theta ), $$ 
and the equality holds if and only if 
$r_p(y)= r(y), l_p(y)=l(y)$, almost everywhere.
\end{proposition}

Clause I is from the weak law of large numbers. To prove Clause II, we first state the following Lemma \citep[adapted from Theorem 6 in][]{gneiting2007strictly}, which addresses a single distribution (no dependence on $y$).

\begin{lemma}
Let $p$ be a continuous distribution on $\Theta=\R$, and $p(\theta)>0$ for any $\theta$.
    If $s$ is a strictly increasing function, and 
    given a fixed confidence-level $\alpha\in (0,1)$, 
    then the scoring rule
$$S(r; \theta) = \alpha s(r) + (s(\theta) - s(r))\mathds{1}(\theta \leq r).$$
is proper for predicting the quantile of $p$ at level $\alpha$.    
\end{lemma} 
\begin{proof}
The proof of the lemma is also adapted from \citet{gneiting2007strictly}.
    Let $\mu$ be the unique $\alpha$-quantile of $p$. For any $r < \mu$, 
    \begin{align*}
       \E_{\theta\sim p(\theta)} S(\mu; \theta) -  \E_{\theta\sim p(\theta)}  S(r; \theta)
       &= \int_{r}^{\mu} s(\theta)p(\theta) d\theta + s(r)P(r) - \alpha s(r) \\
       &>  s(r)(p(\mu) - p(r)) + s(r)p(r) -\alpha s(r)\\
        &= 0.
    \end{align*}
    Likewise,  for any $r < \mu$,   $\E_{\theta\sim p(\theta)} S(q; \theta) -  \E_{\theta\sim p(\theta)}  S(r; \theta) >0.$
\end{proof}

Let $s(x)=x$  and apply this lemma  twice, then for any distribution $p$ on $\Theta=\R$, and a fixed confidence level $\alpha$, suppose $\mu_l, \mu_r$ are the $\alpha/2$ and $1-\alpha/2$  quantile of $p$, 
 it is clear that 
 $$\E_{\theta\sim p(\theta)} U(\mu_r, \mu_l \theta) -  \E_{\theta\sim p(\theta)}  S(\mu_r^\prime, \mu_l^\prime; \theta) \leq  0.
 $$
 The equality holds if and only $\mu_r^\prime=\mu_r, \mu_l^\prime=\mu_l$.

To prove the second clause of Prop.~\ref{thm:interval},  note that 
$\E_{p(\theta,y)} U( r_p(y), l_p(y),  \theta ) -
\E_{p(\theta,y)} U( r(y), l(y),  \theta )
=\E_{y}   \left(\E_{p(\theta|y)} U( r_p(y), l_p(y),  \theta )- 
\E_{p(\theta|y)} U( r(y), l(y),  \theta )\right).$
Because for any given $y$,  $\E_{p(\theta|y)} U( r_p(y), l_p(y),  \theta )- 
\E_{p(\theta|y)} U( r(y), l(y),  \theta ) \leq 0$ 
due to the previous lemma, then 
$
\E_{p(\theta,y)} U( r_p(y), l_p(y),  \theta ) -
\E_{p(\theta,y)} U( r(y), l(y),  \theta )\leq 0 $, and the equality holds if any only if $r_p(y)=  r(y), l_p(y)= l(y)$ almost everywhere with respect to $p(y)$.

\subsection{Moment stacking (Prop.~\ref{thm:moment})}
\begin{proposition}\label{thm:moment}
For any $y$, let $\mu_k(y)$ and $V_k(y)$ be the mean and covariance of the $k$-th approximate posterior  $q_k(\theta|y)$.
Given a weight $\w$, let  $V^*_y(\w)$
and $\mu^*_y(\w)$ be the covariance and mean in the $\w$-mixed posterior $\sum_{k=1}^K q_k(\theta|y)$, as defined in subsection~\ref{sec:Moment}. Then for any fixed $\w$ and $N \to \infty$, 
$$\frac{1}{N}\sum_{n=1}^N \left(\log \mathrm{det} V^*_n(\w) + ||\mu^*_n(\w) - \theta_n||^2_{(V^*_n(\w))^{-1}}\right) = 
\E_{p(\theta, y)}\left(
\log \mathrm{det} V^*_y(\w) + ||\mu^*_y(\w) - \theta||^2_{(V^*_y(\w))^{-1}}\right) + o_p(1). 
$$
Let $\mu(y)$ and $V(y)$ be the true mean and variance of the posterior $p(\theta|y)$, then for any $\w$,
$$
\E_{p(\theta, y)}\left(
\log \mathrm{det} V(y) + ||\mu(y) - \theta||^2_{V^{-1}(y)}\right) 
\leq 
\E_{p(\theta, y)}\left(
\log \mathrm{det} V^*_y(\w) + ||\mu^*_y(\w) - \theta||^2_{(V^*_y(\w))^{-1}}\right), 
$$
and the equality holds if and only if 
$V^*_y(\w)= V(y) $ and $V^*_y(\w)=  \mu(y)$ almost everywhere with respect to p(y), if attainable.
\end{proposition}

The proof is very similar to Prop.~\ref{thm:interval}, requiring one application of the WLLN, and to verify that the underlying score is proper. We omit the details here.

The next proposition states that any mixture of a list of pointwisely under-confident approximations will remain under-confident.  
\begin{proposition}\label{thm:under}
 For a fixed $y$, if for any $k$, $$\mathrm{Var}_{q_k}(\theta|y) \geq  \mathrm{Var}_{p}(\theta|y),$$
 then for any weight $\w \in \s_K$, the variance in the mixture is always under-confident:
$k$, $$\mathrm{Var}_{p_{\w}^{\mix}}(\theta|y) \geq  \mathrm{Var}_{p}(\theta|y).$$
 \end{proposition}
 \begin{proof}
     Use the law of total variance, for any fixed $\w$,
\begin{align*}
    \mathrm{Var}_{p_{\w}^{\mix}}(\theta|y) &=  \mathrm{Var} \E_{q_k}(\theta|y) + \E \mathrm{Var}_{q_k} (\theta|y)\\
    &\geq   \E \mathrm{Var}_{q_k} (\theta|y)\\
     &=   \sum_{k=1}^K w_k \mathrm{Var}_{q_k} (\theta|y)\\
     &\geq \min_{k} \mathrm{Var}_{q_k} (\theta|y)\\
&\geq   \mathrm{Var}_{p} (\theta|y).\\
\end{align*}
 \end{proof}

\section{\NoCaseChange{Practical Implementation}}\label{app:practical}

\subsection{Smooth approximation of indicator functions}\label{app:smooth}

We approximate indicators functions using an infinitely differentiable approximation of the Heaviside step function ${H_\varepsilon(x) = (1 + \exp(-2\varepsilon x))^{-1}}$ for a given $\varepsilon$ value. In practice, we select an $\varepsilon$ value that is sufficiently small compared to typical evaluation points of $H_\varepsilon$.

In particular, in the context of Sect.~\ref{sec:cali}, we choose $\varepsilon = 1/100$. In the context of Sect.~\ref{sec:interval}, we choose $\varepsilon = (\min_n (r_n - l_n) / 1,000$ where $\min_n (r_n - l_n)$ is the minimum interval length over the training set.

\subsection{Quasi Monte Carlo sampling from the stacked posterior}\label{app:sample}
Given weights ${w_k}$, and $K$ simulation draws
$\{\theta_{ks}\} \sim q_{k}(\cdot)$, $k=1,\cdot, K$, the goal is  to draw samples from the stacked inference $\sum_{k=1}^K w_k q_k(\cdot)$. 
We first draw a fixed-sized $S^\star_k = \lfloor Sw_k\rfloor$ sample randomly without replacement from the $k$-th inference, 
and then sample the remaining $S -\sum_{k=1}^KS^\star_k$ samples without replacement with the probability proportional to $w_k - S^\star_k/S$ from inference $k$.

\subsection{Closed-form expression for the rank-based  integral in Eq.~\eqref{eq:rank_obj}}


In the context of Sect.~\ref{sec:cali}, we consider $N$ i.i.d. rank samples $r_1, \dots, r_N$ and their corresponding empirical CDF $\hat{F}_{r, N}(t) = \frac{1}{N}\sum_{n=1}^N \mathds 1 (r_n \leq t)$. We derive a closed-form expression of a Cramér–von Mises-type distance between $\hat{F}_{r, N}$ and the CDF of a uniform distribution $F_{\mathcal{U}(0, 1)}(t) = t 1_{t\in [0, 1]}$:
\begin{align*}
\int_0^1|\hat{F}_{r, N}(t) - t|^2\,dt &=  \int_0^1\hat{F}_{r, N}(t)^2\,dt - 2\int_0^1 t\hat{F}_{r, N}(t) \,dt + \frac{1}{3}, \\
&= \frac1{N^2}\sum_{i, j = 1}^N \int_0^1 \mathds 1 (r_i \leq t)\mathds 1 (r_j \leq t) \,dt - \frac{2}{N}\sum_{i=1}^N \int_0^1 t\,\mathds 1 (r_i \leq t) \,dt + \frac{1}{3}, \\
&= \frac1{N^2}\sum_{i, j = 1}^N (1 - \max(r_i, r_j)) - \frac{1}{N}\sum_{i=1}^N (1 - r_i^2) + \frac{1}{3}, \\
&= \frac{1}{N}\sum_{i=1}^N r_i^2 - \frac1{N^2}\sum_{i, j = 1}^N\max(r_i, r_j) + \frac{1}{3}.
\end{align*}

\section{\NoCaseChange{Experiment Details}}\label{app:exp}
\subsection{Toy example: hybrid stacking in  Gaussian posteriors}

In Section \ref{sec:theory}, we design a true posterior inference  
$\theta|y \sim \n(y, 1)$. We consider four manually corrupted posterior inferences,\\
\hspace*{1.65em} Inference 1: ~$\theta|y \sim \n(y+1, 1)$,\\
\hspace*{2em}Inference 2: ~$\theta|y \sim \n(y-1, 1)$,\\
\hspace*{2em}Inference 3: ~$\theta|y \sim \n(y, 0.56)$,\\
\hspace*{2em}Inference 4: ~$\theta|y \sim \n(y+0.5, 2.45)$.\\
They are designed in such a way that the KL divergence between the true posterior and each of the four approximate inferences is roughly the same.

In hybrid stacking, we maximize the sum of log density and rank-calibration:
$$\max_{\w\in\s_K}\left(
{\sum_{n=1}^N \log  \sum_{k=1}^K w_k q_k(\theta_n| y_n)  }
-\lambda\left(
{  \left(\frac{1}{N}\sum_{n=1}^N \sum_{k=1}^K {\log(w_k r_{kn})} +1\right)^2}  + \left(\frac{1}{N}\sum_{n=1}^N \sum_{k=1}^K {(w_k r_{kn})} - \frac{1}{2}\right)^2\right)\right).
$$

We use $\lambda=100$, because the calibration error is of a smaller scale. 
Though it is straightforward to further tune $\lambda$
via standard cross-validation, we keep a default value without tuning.



\subsection{Hyperparameters in SBI Benchmark and SimBIG}
In Section~\ref{sec:exp}, we have demonstrated our stacking on three models from the SBI Benchmark: ``Two Moons'', ``simple likelihood and complex posterior'', and SIR model. In addition, we applied our method to a cosmological inference task ``SimBIG''.

For each inference task, we run a large number of neural posterior inferences, constructed by varying the hyperparameters in the normalizing flow on a grid. Table~\ref{table:sweep_conf} summarizes the range of hyperparameters we used.

\begin{table*}
\centering
    \def\arraystretch{1.5}
    \begin{tabular}{c|cc|cc|c}
        \multicolumn{1}{c|}{\multirow{2}{*}{Hyperparameter}} & \multicolumn{2}{c|}{\multirow{1}{*}{Minimum value}} & \multicolumn{2}{c|}{\multirow{1}{*}{Maximum value}} & \multicolumn{1}{c}{\multirow{2}{*}{Distribution}}\\
        \cline{2-5}
        & \texttt{sbibm} & SimBIG & \texttt{sbibm} & SimBIG & \\
        \hline
        Nb. of MAF layers & 3 & 5 & 8 & 11 & uniform \\
        Nb. of MLP hidden units & 32 & 256 & 256 & 1024 & log-uniform \\
        Nb. of MLP layers & 2 & 2 & 4 & 4 & uniform \\
        MLP Dropout prob. & 0.0 & 0.1 & 0.2 & 0.2 & uniform \\
        Batch size & 20 & 20 & 100 & 100 & uniform \\
        Learning rate & 1e-5 & 5e-6 & 1e-3 & 5e-5 & log-uniform
    \end{tabular}
    \caption{Constraints for the random selection of the hyperparameters of the neural posterior estimators.}
    \label{table:sweep_conf}
\end{table*}

\subsection{Additional experiments using neural spline flows} 
In the main paper we have tested stacking on Masked Autoregressive Flow.  Here we run stacking on five SBI benchmark tasks (two moons, SLCP, SIR, Gaussian mixture and Lotka-Volterra). In each task, we run  $K=50$  neural spline flows to approximate the Bayesian posterior, and run posterior stacking these flows. 
Table \ref{table2} demonstrates the gain of the mixture stacking and the moment stacking where we evaluate the log predictive density or the moment errors in each task. Our stacking approach outperforms uniform weighting and selection.

\begin{table}
    \centering
 \begin{tabular}{ c || c c c || c c c }
\multicolumn{1}{c||}{\multirow{3}{*}{Task}} & \multicolumn{3}{c||}{\multirow{2}{*}{\parbox{3.5cm}{\centering Mixture for KL [Log~Pred.~Density] $\uparrow$}}} & \multicolumn{3}{c}{\multirow{2}{*}{\parbox{3cm}{\centering Moments Stacking [Moments Error] $\downarrow$}}}\\
&&&&&& \\
\cline{2-7}
& Best & Unif. & Stacked  & Best & Unif. & Stacked\\
\hline
{ Two Moons} & \mes{3.57}{.01} & \mes{3.54}{.01} & \mes{\textbf{3.66}}{.01} & \mes{-1.72}{.02} & \mes{-1.72}{.02} & \mes{\textbf{-1.73}}{.02} \\
{ SLCP} & \mes{-4.44}{.03} & \mes{-4.38}{.02} & \mes{\textbf{-4.06}}{.03} & \mes{0.86}{.02} & \mes{0.99}{.01} & \mes{\textbf{0.75}}{.01} \\
{ SIR} & \mes{7.78}{.02} & \mes{7.77}{.02} & \mes{\textbf{7.88}}{.02} & \mes{-7.02}{.02} & \mes{-6.99}{.01} & \mes{\textbf{-8.49}}{.02} \\ 
{ Gauss. Mixture} & \mes{-1.24}{.02} & \mes{-1.16}{.02} & \mes{\textbf{-1.13}}{.02} & \mes{0.28}{.01} & \mes{0.30}{0.01} & \mes{\textbf{0.26}}{.01} \\
{ Lotka-Volterra} & \mes{12.94}{.02} & \mes{12.69}{.02} & \mes{\textbf{13.47}}{.02} & \mes{-6.58}{.03} & \mes{-5.52}{.01} & \mes{\textbf{-6.91}}{.01} 
\end{tabular}
\caption{We run stacking on five benchmark tasks, each  with the $K=50$ neural spline flows. The gray number indicates the standard error.  }\label{table2}
\end{table}

\end{document}